\documentclass[12pt,fleqn,leqno]{article}
\usepackage{amsfonts}
\usepackage{amsmath,amssymb}
\usepackage{amsthm,color}
\usepackage[dvips]{graphicx}
\usepackage{setspace}

\allowdisplaybreaks

\numberwithin{equation}{section}
\makeatletter
\renewcommand{\section}{\@startsection{section}{1}{0pt}{20pt}{6pt}{\large\bf}}
\renewcommand{\@seccntformat}[1]{\csname the#1\endcsname.\ }

\def\footnoterule{\kern -3pt \hrule width 2.7 true cm \kern 2.6pt}

\def\ni{\noindent}
\def\vs{\vspace}
\def\hs{\hspace}
\def\EE{\mathsf E}

\def\cC{\mathcal C}
\def\cE{\mathcal E}
\def\cD{\mathcal D}

\def\L{I\!\!L}
\def\wt{\widetilde}

\newcommand{\p}{\! +\! }
\newcommand{\m}{\! -\! }
\newtheorem{theorem}{Theorem}[section]
\newtheorem{lemma}[theorem]{Lemma}

\newtheorem{definition}[theorem]{Definition}
\newtheorem{remark}[theorem]{Remark}

\begin{document}

\title{\textbf{American Options with Discontinuous Two-Level Caps}}
\author{Jerome Detemple \& Yerkin Kitapbayev}
\date{19 July 2017}
\maketitle


{\par \leftskip=2.6cm \rightskip=2.6cm \footnotesize This paper examines the
valuation of American capped call options with two-level caps. The structure
of the immediate exercise region is significantly more complex than in the
classical case with constant cap. When the cap grows over time, making
extensive use of probabilistic arguments and local time, we show that the
exercise region can be the union of two disconnected set. Alternatively, it
can consist of two sets connected by a line. The problem then reduces to the
characterization of the upper boundary of the first set,
which is shown to satisfy a recursive integral equation. When the cap
decreases over time, the boundary of the exercise region has piecewise
constant segments alternating with non-increasing segments. General
representation formulas for the option price, involving the exercise
boundaries and the local time of the underlying price process, are derived.
An efficient algorithm is developed and numerical results are provided.\par}





\vspace{-18pt}

\vspace{-18pt}

\section{Introduction}


Caps are attractive provisions found in a variety of derivative contracts.
For an investor seeking to establish a short position, a cap limits the
exposure to unfavorable price movements. For a buyer, it reduces the price,
hence the cost of entering a long position. Products with caps can be found
on organized exchanges, in OTC markets, or appear as components of financial
packages offered by firms to raise capital. A common example is the CAPS
contract, introduced in 1991 by the CBOE and traded to the end of the 90s.
This option, which is written on the S\&P100 and S\&P500 indices, has gains
limited to $30$-point moves in the underlying index. It also closes
automatically, hence has an automatic exercise provision, once the change in
the underlying index hits the cap. Another example is the Mexican
Index-Linked Euro Security (MILES), which traded in the early 90s on the
Luxembourg stock exchange. The MILES is an American capped option on the
Mexican stock index, the Bolsa Mexicana de Valores. It has a initial vesting
period during which exercise is prohibited, followed by a period during
which exercise is at the option of the holder, but with gains limited by the
cap provision. It is thus a special case of an American capped option with
two-level cap. Nowadays, caps appear as features that can be incorporated in
FLEX options offered by the CBOE. FLEX options are fully customizable
products written on indices or equities. They can be European- or
American-style. Bull and Bear spreads are examples of FLEX options with caps.

The valuation of capped option has been a topic of long-standing interest in
academic and practitioner circles. Early studies dealing with the valuation
of European-style capped options include Boyle and Turnbull (1989). Capped
options with automatic exercise at the cap are discussed by Flesacker (1992)
and Chance (1994), among others. The case of American-style capped options
is examined by Broadie and Detemple (1995). They consider both constant and
growing caps. Using dominance arguments, they show that the optimal exercise
policy for a capped call with maturity date $T$ and constant cap $L$ is the
first hitting time of $B\wedge L$, the minimum of the cap $L$ and the
optimal exercise boundary $B=\left( B\left( t\right) \right) _{t\in \left[
0,T\right] }$ of an uncapped, but otherwise identical option. They deduce
several representations of the capped option price. For a cap that grows at
a constant rate $g$, they show that the optimal policy lies in a 3-parameter
family. They characterize the optimal parametric policy and find the
associated option value. In a recent study, Qiu (2015) analyzes American
capped strangle options and exploits the local time-space formula to derive
an early exercise premium (EEP) representation. The pair of optimal exercise
boundaries is characterized as the unique solution to a system of coupled
integral equations.

This paper is devoted to the pricing of an American capped call option with
a two-level cap $L_{1}1_{\tau <T_{1}}+L_{2}1_{T_{1}\leq \tau \leq T_{2}}$,
where $L_{1},L_{2}$ are arbitrary levels such that $L_{1},L_{2}>K$. Such a
cap structure provides more flexibility from the point of view of an issuing
firm seeking to raise capital. The case $L_{1}<L_{2}$ is especially relevant
for the success of an issue, as the growing upside potential of the contract
raises its attractiveness for potential buyers. At the same time, the
two-tiered cap still limits the exposure of the firm, making it interesting
for the issuer. Nevertheless, for completeness, we also consider the case $%
L_{2}<L_{1}.$

It is clear that on the interval $[T_{1},T_{2}]$, the pricing problem
reduces to the problem with single cap $L_{2}$ for which the immediate
exercise region $\mathcal{E}_{2}$ was characterized by Broadie and Detemple
(1995). Thus, $\mathcal{E}_{2}=$ $\left\{ \left( S,t\right) \in \mathbb{R}%
^{+}\times \left[ T_{1},T_{2}\right] :S\geq B\left( t\right) \wedge
L_{2}\right\} $. We make use of the local time-space calculus of Peskir
(2005) and provide an alternative representation for the option price using
the European capped option as the benchmark for the early exercise premium
formula. Due to the non-smoothness of the payoff function in the exercise
region at the cap $L_{2}$, a local time term appears in the pricing formula.
However, this representation is found to be efficient and we exploit it
extensively for the analysis of the optimal exercise strategy on $[0,T_{1})$%
. Remarkably, this representation also provides an alternative proof of the
result of Broadie and Detemple (1995). Namely, it shows that the part of the
boundary below the cap satisfies the same integral equation as the boundary
of the American uncapped call option (see Remark 2.6).

We then consider the pricing problem over the interval $[0,T_{1})$. There
are three cases of interest. In the first case, $L_{1}<B\left( T_{1}\right)
\wedge L_{2}$. We show that the immediate exercise region on $[0,T_{1})$
takes the form%
\begin{equation*}
\hspace{2pc} \mathcal{E}_{1}=\left\{ \left( S,t\right) \in \mathbb{R}%
^{+}\times \left[ 0,t^{1}\right] :L_{1}\leq S\leq B^{L,1}(t)\right\}
\end{equation*}%
where $B^{L,1}=\left( B^{L,1}\left( t\right) \right) _{t\in \left[ 0,t^{1}%
\right] }$ is an upper boundary and $t^{1}<T_{1}$ is an endogenous time. It
has three unusual features: (i) if $t^{0}\equiv T_{1}-\frac{1}{r}\log
((L_{2}-K)/(L_{1}-K))\geq 0$, then $B^{L,1}(t)=+\infty $ for $t\in \lbrack
0,t^{0}]$ and $B^{L,1}(t^{0}+)=+\infty $; (ii) there exists a time $%
t^{0}\vee 0\le t^{1}<T_{1}$ such that $\left( S,t\right) \notin \mathcal{E}%
_{1}$ for any $S>0$ and $t\in \left( t^{1},T_{1}\right) $; (iii) there
exists a time $t^{0}\vee 0\le T_{0}\le t^{1}$ such that $B^{L,1}(t)=L_{1}$
for$\;t\in \lbrack T_{0},t^{1}]$. We characterize the times $\left(
t^{0},T_{0},t^{1}\right) $ and show that the upper boundary $B^{L,1}$ solves
a recursive integral equation on $[t^0\vee 0,T_0]$.

In the second case, $B\left( T_{1}\right) \leq L_{1}<L_{2}$. The immediate
exercise region on $[0,T_{1})$ becomes%
\begin{align*}
\hspace{2pc}\mathcal{E}_{1}=& \left\{ \left( S,t\right) \in \mathbb{R}%
^{+}\times \left[ 0,T_{1}\right) :L_{1}\leq S\leq B^{L,1}(t)\right\} \\
&\cup \left\{ \left( S,t\right) \in \mathbb{R}^{+}\times \left[ t^{\ast
},T_{1}\right) :B(t)\leq S\leq L_{1}\right\}
\end{align*}%
where $t^{\ast }$ is the unique solution to the equation $B\left( t\right)
=L_{1}$ and the upper boundary $B^{L,1}=\left( B^{L,1}\left( t\right)
\right) _{t\in \left[ 0,T_{1}\right) }$ is defined over the whole interval $%
\left[ 0,T_{1}\right) $. Moreover, there exists $T_0<T_1$ such that $%
B^{L,1}(t)=L_1$ for $t\in[T_0,T_1)$. The main differences with the previous
case are that immediate exercise (i) is always optimal when $S=L_{1}$ (thus $%
t^{1}=T_{1}$) and (ii) is optimal below $L_{1}$ when $B(t)\leq S\leq L_{1}$.
We derive the corresponding recursive integral equation for $B^{L,1}$ on $%
[t^0\vee 0,T_0]$.

In the last case, when $L_{1}>L_{2}$, we assume that the cap is
left-continuous. We impose this condition in order to avoid technical
difficulties associated with the discontinuity of the option price at $T_{1}$
from the left and the possible non-existence of an optimal exercise time.
Moreover, the prices of the options with left- and right-continuous caps
coincide prior to $T_{1}$. For this case, we show that the immediate
exercise region is%
\begin{equation*}
\mathcal{E}=\left\{ \left( S,t\right) \in \mathbb{R}^{+}\times \left[ 0,T_{1}%
\right] :S\geq B^{L,1}\left( t\right) \right\} \cup \left\{ \left(
S,t\right) \in \mathbb{R}^{+}\times \left( T_{1},T_{2}\right] :S\geq
B(t)\wedge L_{2}\right\}
\end{equation*}%
where $B^{L,1}=\left( B^{L,1}\left( t\right) \right) _{t\in \left[ 0,T_{1}%
\right] }$ is an endogenous non-increasing boundary defined over $\left[
0,T_{1}\right] $ satisfying the boundary condition $B^{L}(T_{1}-)=\max(rK/%
\delta,L_{2}\wedge B(T_{1}))$ and $B^{L}(T_{1})=L_{2}\wedge B(T_{1})$. In
this instance, immediate exercise above the cap $L_{1}$ is always optimal on
$[0,T_1)$. Immediate exercise below the cap is also optimal if $%
B^{L,1}(t)\le S \le L_1$. We derive the recursive integral equation for $%
B^{L,1}$ on $[0,T_1)$.

In all cases, the immediate exercise region over the whole interval $%
[0,T_{2}]$ is the union of the two sets, $\mathcal{E}_{1}=\mathcal{E}%
_{1}\cup \mathcal{E}_{2}$. In the first case described above, the sets in
question are disconnected, i.e., they are separated by a region in which
continuation is optimal. For each case, we derive a general representation
formula for the option price, involving the relevant boundary and the local
times of the underlying price process at the cap and the boundary. We also
develop an efficient algorithm for computing the exercise boundary and the
option price, and provide numerical results.

The main difficulties for pricing this option and determining the optimal
exercise policy, compared to the standard American-style derivatives, are
the non-smoothness of the payoff with respect to underlying asset price in
the exercise region and the time-discontinuity of the payoff. Although the
first issue can be tackled as in aforementioned papers, the second feature
induces a complicated structure of the exercise region and prevents the
application of standard methods. This paper provides probabilistic
arguments, making extensive use of dominance relations and local time to
tackle these types of difficulties. Other potential applications of our
approach include finite horizon impulse control problems in the presence of
execution delays and decision lags, which are common in algorithmic trading
and other decision-making processes. For related studies and theoretical
background, we refer to Palczewski and Stettner (2010), among others, who
showed regularity of the value functions of optimal stopping problems with
time-discontinuous payoffs under a quite general class of Markov processes.

The paper is organized as follows. Section 2 formulates the problem,
revisits the results of Broadie and Detemple (1995) and provides an
alternative pricing formula for the constant cap option problem. Section 3
contains the main result of the paper. It describes the optimal exercise
policy for the three cases of interest. Sections 4, 5 and 6 consider each
case in isolation and prove the optimality of the exercise policy announced
in Section 3. They also provide the relevant early exercise premium formulas
for the option price. Finally, Section 7 describes a step-by-step algorithm
for computation and presents numerical results.\vspace{6pt}

\section{Problem Formulation and Preliminary Results}

Consider the probability space $(\Omega ,\mathcal{F},Q)$ where $Q$ is the
risk-neutral measure. Let $W$ be a $Q$-standard Brownian motion with natural
filtration $(\mathcal{F}_{t})_{t\geq 0}$. We assume the standard financial
market with an underlying asset price $S$ that follows a geometric Brownian
motion under the risk-neutral measure $Q$
\begin{equation}
\hspace{6.0905pc}dS_{t}=S_{t}\left( \left( r-\delta \right) dt+\sigma
dW_{t}\right)
\end{equation}%
where $r>0$ is the interest rate, $\delta \geq 0$ the dividend yield, and $%
\sigma >0$ the volatility parameter.

An American capped call option has exercise payoff $\left( S_{\tau }\wedge
L(\tau)-K\right) ^{+}$ at exercise time $\tau $ where $L=\left( L\left(
t\right) \right) _{t\in \left[ 0,T_{2}\right] }$ is the cap. In this paper,
we focus on capped options with two-level caps, i.e., such that%
\begin{equation}
\hspace{6.1133pc}L(\tau)=L_{1}1_{\tau <T_{1}}+L_{2}1_{T_{1}\leq \tau \leq
T_{2}}
\end{equation}%
with $T_{1}<T_{2}$. The cap $L(t)$ is right-continuous with jump at the
intermediate date $T_{1}$ if $L_{1}\neq L_{2}$. If $L_{1}>L_{2}$, we
restrict attention to the left-continuous version of this cap, namely $%
L(\tau)=L_{1}1_{\tau \leq T_{1}}+L_{2}1_{T_{1}<\tau \leq T_{2}}$. The
strike price is $K$ and the maturity date $T_{2}$. To avoid trivial cases we
assume $L_{1},L_{2}>K$.

Standard arguments (e.g., Karatzas (1988)) can be invoked to write the price
function as%
\begin{equation}
\hspace{6.0905pc}C^{A,L}(S,t)=\sup\limits_{t\leq \tau \leq T_{2}}\mathsf{E}%
_{t}\left[ e^{-r(\tau -t)}\left( S_{\tau }\wedge L(\tau)-K\right) ^{+}%
\right]   \label{problem-1}
\end{equation}%
at time $0\leq t\leq T_{2}$ and given the stock price $S>0$, where the
supremum is taken over all stopping times $\tau \in \lbrack t,T_{2}]$ w.r.t.
the filtration $(\mathcal{F}_{t})_{t\geq 0}$. Throughout the paper, we use
the notation $\mathsf{E}_{t}\left[ \cdot \right] \equiv \mathsf{E}[\,\cdot
\,|S_{t}=S]$ and $\mathsf{E}$ is the expectation under $Q$. We also define
the infinitesimal generator of $S$ as
\begin{equation}
\hspace{6.0905pc}I\!\!L_{S}=(r-\delta )S\frac{\partial }{\partial S}+\frac{%
\sigma ^{2}}{2}S^{2}\frac{\partial ^{2}}{\partial S^{2}}.
\end{equation}

We note that it is not true in general that the value functions of optimal
stopping problems with discontinuous payoffs are continuous and that optimal
stopping rules have the standard form. However, if $L_{1}<L_{2}$ we can
apply Theorem 3.1 iii) of Palczewski and Stettner (2010) to obtain that $%
C^{A,L}$ is continuous and the optimal exercise time is given as%
\begin{equation}
\hspace{6.0678pc}\tau ^{\ast }(t)=\inf \{u\geq
t:C^{A,L}(S_{u},u)=(S_{u}\wedge L(u)-K)^{+}\}\wedge T_{2}.
\end{equation}%
for $t\leq T_{2}$. The case $L_{1}>L_{2}$ is more complicated as the value
function becomes discontinuous at $T_{1}$ from the left and we may have
non-existence of an optimal stopping time. To avoid these issues, we will
consider the left-continuous cap $L(\tau)=L_{1}1_{\tau \leq
T_{1}}+L_{2}1_{T_{1}<\tau \leq T_{2}}$ when $L_{1}>L_{2}$ and show the
connection between two versions (see Section 6).

Let $B=(B(t))_{t\in \lbrack 0,T_{2}]}$ denote the optimal exercise boundary
of an uncapped call option with identical strike $K$ and maturity date $T_{2}
$, i.e., with exercise payoff $\left( S_{\tau }-K\right) ^{+}$. In the case $%
L_{1}<L_{2}$, following Broadie and Detemple (1995), we define $t^{\ast }$
as follows%
\begin{equation*}
\hspace{0pc}\left\{
\begin{array}{l}
t^{\ast }=s,\quad \text{if }B(s)=L(s)\text{ for some }s\in \left[ 0,T_{2}%
\right]  \\
t^{\ast }=0,\quad \text{if }B(\cdot )<L(\cdot)\text{ on }\in \left[ 0,T_{2}%
\right]  \\
t^{\ast }=T_{2},\quad \text{if }B(T_{2}-)>L_{2} \\
t^{\ast }=T_{1},\quad \text{if }L_{1}<B\left( T_{1}\right) <L_{2}.%
\end{array}%
\right.
\end{equation*}%
Thus, $t^{\ast }$ is the unique solution to the equation $B(s)=L(s)$, if
such a solution exists. Otherwise, $t^{\ast }$ is at a corner of the
interval $\left[ 0,T_{2}\right] $. It is clear that for $t\in \lbrack
T_{1},T_{2}]$, the optimal exercise rule follows the strategy for the capped
option $C^{A,L_{2}}$ with single cap $L_{2}$. This problem has been solved
in Broadie and Detemple (1995) and, in Theorems 2.1 and 2.2 below, we
revisit their results.

\begin{theorem}
The optimal exercise boundary of the capped call option \eqref{problem-1} on
$\left[ T_{1},T_{2}\right] $ is $B^{L,2}=B\wedge L_{2}$ so that the optimal
exercise time in \eqref{problem-1} is given by
\begin{equation}
\hspace{6.0225pc}\tau ^{\ast }(t)=\inf \left\{ u\geq t:S_{u}\geq
B^{L,2}(u)\right\} \wedge T_{2}  \label{tau-star}
\end{equation}%
for $t\in \lbrack T_{1},T_{2}]$. The price of the capped call %
\eqref{problem-1} is then
\begin{equation}
\hspace{6.0225pc}C^{A,L_{2}}\left( S,t\right) =\mathsf{E}_{t}\left[
e^{-r(\tau ^{\ast }(t)-t)}\left( S_{\tau ^{\ast }(t)}\wedge L_{2}-K\right)
^{+}\right]
\end{equation}%
for $S>0$ and $t\in \lbrack T_{1},T_{2}]$.
\end{theorem}

\begin{figure}[t]
\begin{center}
\includegraphics[scale=0.55]{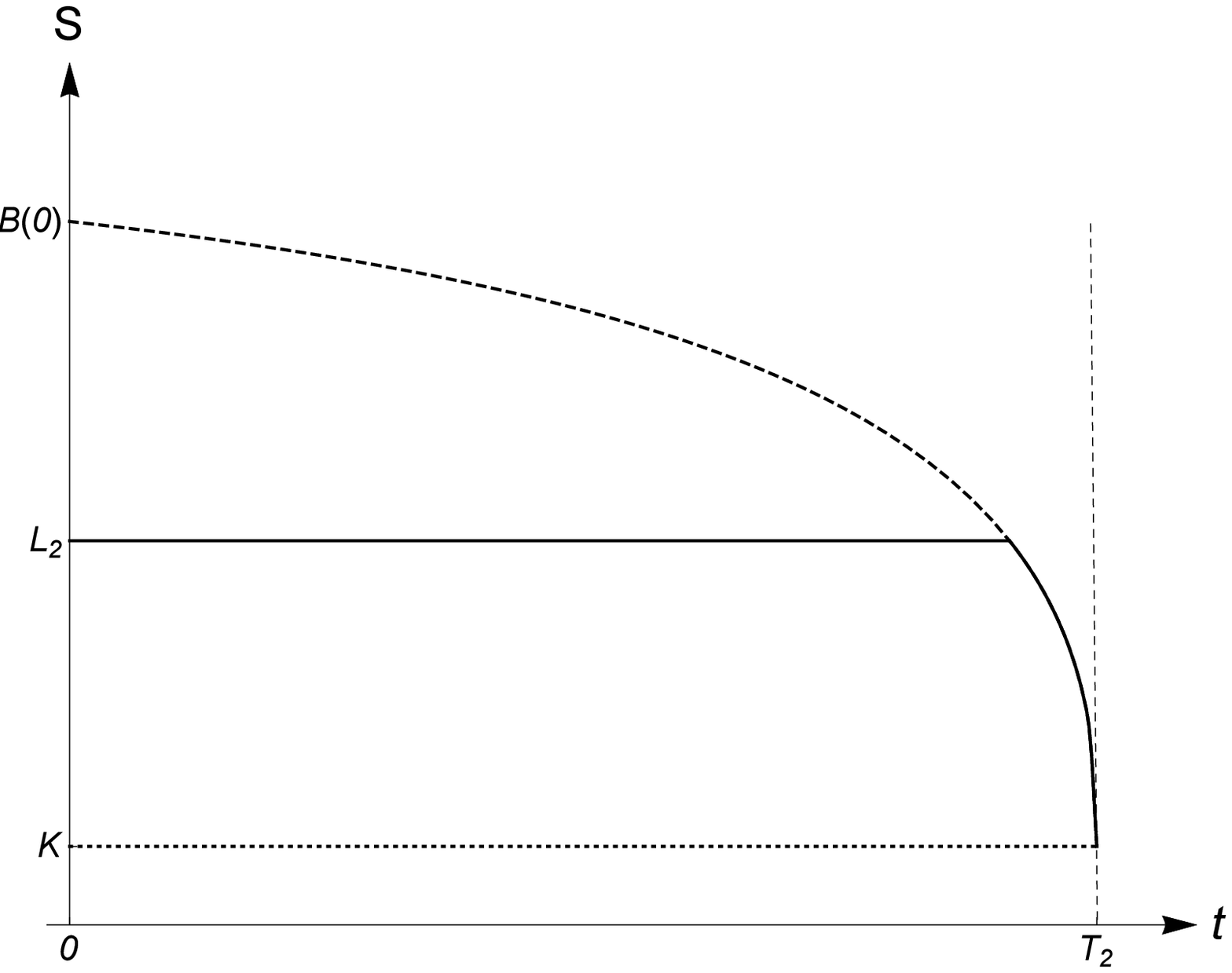}
\end{center}
\par
{}
\par
\leftskip=1.6cm \rightskip=1.6cm {\small \noindent \vspace{-10pt} }
\par
{\small \textbf{Figure 1.} This figure plots the optimal exercise boundary $%
B^{L,2}=B\wedge L_2$ for the single cap option problem. The dashed curve
represents the uncapped exercise boundary $B$. The parameter set is $T_2=4,
K=1, L_2=1.39, r=0.1, \delta=0.1, \sigma=0.3$. }
\par
\vspace{10pt}
\end{figure}

Figure 1 illustrates the optimal exercise boundary $B^{L,2}=B\wedge L_{2}$
for the single cap problem. The optimal exercise region $\mathcal{E}_{2}$
for the problem \eqref{problem-1} when time is restricted to the interval $%
\left[ T_{1},T_{2}\right] $, is then clearly
\begin{equation}
\hspace{6.0451pc}\mathcal{E}_{2}=\left\{ \left( S,t\right) \in \mathbb{R}%
^{+}\times \left[ T_{1},T_{2}\right] :S\geq B^{L,2}(t)\right\} .
\end{equation}%
Using the theorem above, Broadie and Detemple (1995) derived a pricing
formula for $C^{A,L_{2}}\left( S,T_{1}\right) $ by noting that if $t_{\ast
}>T_{1}$ then $S$ either hits $L_{2}$ before $t^{\ast }$ or the capped
option becomes essentially the uncapped call at $t^{\ast }$. If $t^{\ast
}\leq T_{1}$ , then for $S<L_{2}$ the capped option price coincides with the
uncapped option price. An alternative formula was obtained by using a capped
option $C^{ae,L_{2}}$ with automatic exercise at $L_{2}$ as the benchmark
for an early exercise premium representation.

\begin{theorem}
$(i)$ For $t\in[T_1,T_2]$, the price of the capped call option with
discontinuous cap $L$ is equal to the capped call option price with constant
cap $L_2$ so that
\begin{align}  \label{capped-price-1}
C^{A,L}\left( S,t\right) &=C^{A,L_{2}}\left( S,t\right) \\
&=(L_{2}\! -\! K)\mathsf{E}_{t}\left[ e^{-r\left( \tau ^{\ast }(t)-t\right)
}1_{\{\tau ^{\ast }(t)<t^{\ast }\}}\right] +\mathsf{E}_{t}\left[ e^{-r\left(
t^{\ast }\vee t-t\right) }C^{A}\left( S_{t^{\ast }\vee t},t^{\ast}\vee t
\right) 1_{\{\tau^{\ast }(t)\geq t^{\ast }\}}\right]  \notag
\end{align}%
for $S>0$, where $C^{A}$ is the price of a standard American call option
with strike $K$ and maturity $T_{2}$. The second term can be simplified
using the known distribution of $\tau ^{\ast }(t)$ on $[t,t^{\ast }\vee t]$
(i.e., first hitting time of a constant barrier) and is then given in terms
of an integral of $C^{A}$. Lemma \protect\ref{lemma:appendix} in Appendix shows how to compute both expectations. \vspace{6pt} 

$(ii)$ Alternatively, the capped option can be priced as
\begin{equation}
C^{A,L_{2}}\left( S,t\right) =C^{ae,L_{2}}(S,t)+\mathsf{E}_{t}\left[
\int_{t}^{\tau _{L_{2}}}e^{-r(u-t)}(\delta S_{u}\!-\!rK)1_{\{S_{u}\geq
B(u)\}}du\right]  \label{capped-price-2}
\end{equation}%
for $S>0$ and $t\in \lbrack T_{1},T_{2}]$, where $C^{ae,L_{2}}(S,t)=\mathsf{E%
}_{t}\left[ e^{-r\left( \tau _{L_{2}}-t\right) }\left( S_{\tau
_{L_{2}}}\wedge L_{2}-K\right) ^{+}\right] $ is the price of a capped option
$C^{ae,L_{2}}$ with automatic exercise at cap $L_{2}$ and $\tau
_{L_{2}}=\inf \{u\geq t:S_{u}=L_{2}\}\wedge T_{2}$ is the first hitting time
of $L_{2}$ on $[t,T_{2}]$. The option price $C^{ae,L_{2}}$ can be computed
in the closed form.
\end{theorem}

\vspace{6pt}

Formula \eqref{capped-price-1} requires the integration of the standard
American call option price $C^{A}(S,t^{\ast })$ with respect to $S$, whereas
the formula \eqref{capped-price-2} boils down to the integration of the
early exercise premium with respect to the joint distribution of $\tau
_{L_{2}}$ and $S$. In this paper, we will make use of a third version of the
capped option price, which is computationally convenient for our purposes.
It is also based on an early exercise premium decomposition where the
benchmark, in contrast, is given by the European capped option price $%
C^{E,L_{2}}(S,t)=$ $\mathsf{E}_{t}\left[ e^{-r\left( T_{2}-t\right) }\left(
S_{T_{2}}\wedge L_{2}-K\right) ^{+}\right] $. The proof relies on the local
time-space formula on curves (Peskir (2005)). The noteworthy feature of the
American capped option is that the payoff function $\left( S\wedge
L_{2}-K\right) ^{+}$ is not smooth in the exercise region $\mathcal{E}_{2}$
when $S=L_{2}$. This fact implies the presence of a local time term in the
option pricing formula (a similar result is obtained by Qiu (2015) for
American capped strangle options). Figure 2 illustrates the American capped
option price $C^{A,L_{2}}(S,0)$ with the constant cap $L_{2}$. It can be
clearly seen that the local time term brings significant value.

\begin{figure}[t]
\begin{center}
\includegraphics[scale=0.75]{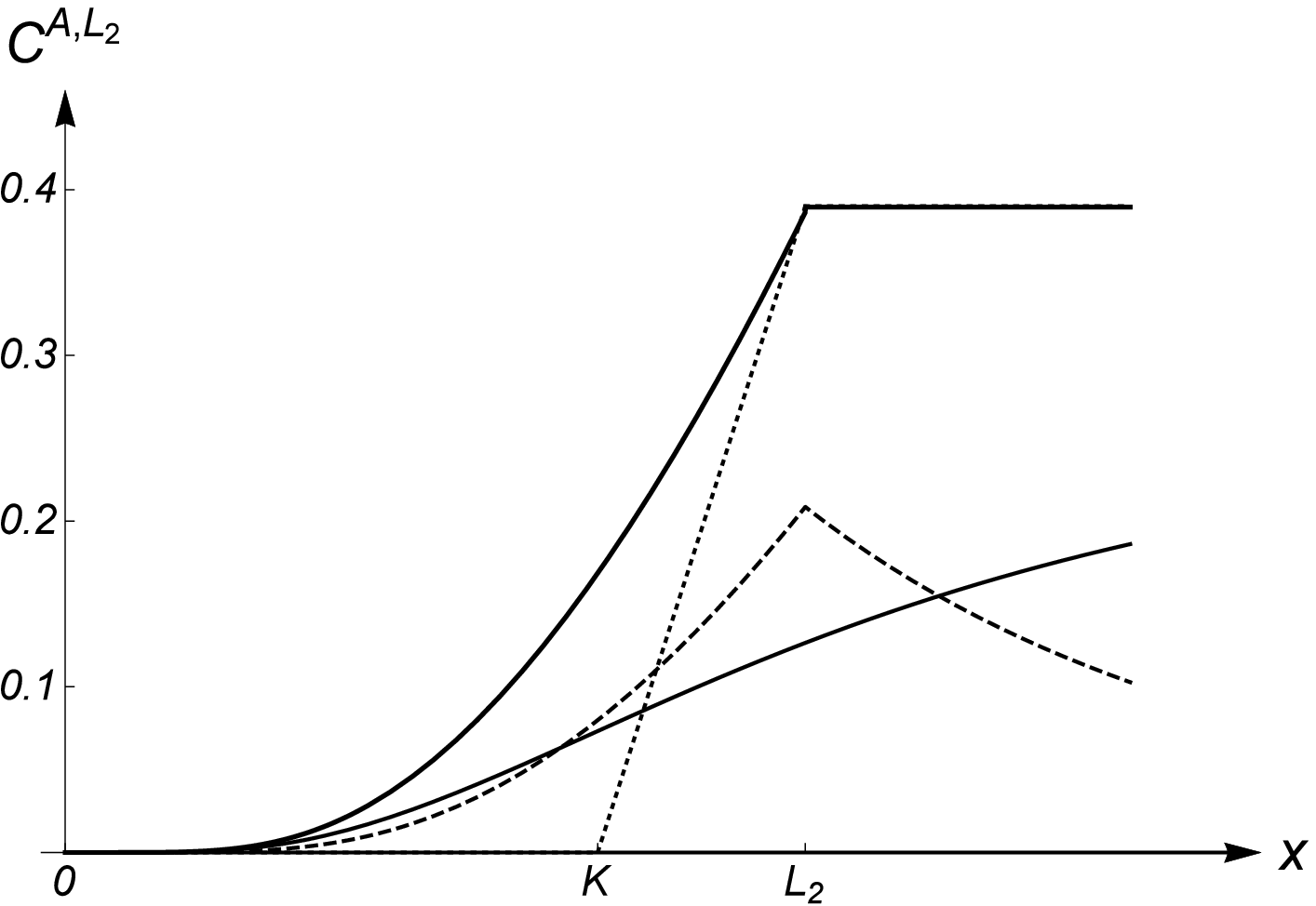}
\end{center}
\par
{}
\par
\leftskip=1.6cm \rightskip=1.6cm {\small \noindent \vspace{-10pt} }
\par
{\small \textbf{Figure 2.} This figure plots the American capped option
price $C^{A,L_{2}}(S,0)$ with the constant cap $L_2$ (thick line). The thin
line represents the European capped option price $C^{E,L_{2}}(S,0)$, the
dashed line represents the local time contribution, the dotted line is the
immediate payoff $(S\wedge L_2-K)^+$. The parameter set is $T_2=4, K=1,
L_2=1.39, r=0.1, \delta=0.1, \sigma=0.3$. It can be seen that the local time
term has substantial value. }
\par
\vspace{10pt}
\end{figure}

\vspace{6pt}


\begin{theorem}
The price of the American capped option $C^{A,L_2}(S,t)$ is
\begin{align}  \label{capped-price-3}
\hspace{3pc} C^{A,L_2}(S,t)=&C^{E,L_2}(S,t) +r(L_2\! -\! K)\mathsf{E}_{t}%
\left[\int_{t}^{T_2} e^{-r(u-t)} 1_{\{S_u\ge L_2\}}du \right] \\
&+\mathsf{E}_{t}\left[\int_{t}^{T_2} e^{-r(u-t)}(\delta S_u\! -\! rK)1_{\{
B(u)\le S_u\le L_2\}}du \right]  \notag \\
&+\frac{1}{2}\mathsf{E}_{t}\left[\int_{t}^{T_2}
e^{-r(u-t)}C^{A,L_{2}}_S(L_2-,u)d\ell_u^{L_2} \right]  \notag
\end{align}
for $t\in[T_1,T_2]$ and $S>0$, where $\ell^{L_2}=(\ell^{L_2}_t)_{t\in
[T_1,T_2]}$ is the local time that the process $S$ spends at $L_2$
\begin{align}  \label{local-time}
\hspace{3pc} \ell^{L_2}_t= Q-\lim_{\varepsilon \downarrow 0}\frac{1}{%
2\varepsilon}\int_{T_1}^{t}
1_{\{L_2-\varepsilon<S_u<L_2+\varepsilon\}}d\left \langle S \right \rangle_u.
\end{align}
\end{theorem}

\begin{proof}
 First, let us define
\begin{align*}
&\cC^{L_2}=\left\{ (S,t) \in \mathbb{R}%
^{+}\times \left[T_{1},T_2\right): S< B^{L,2}(t)\right\}\\
&\cD^{L_2}_1=\left\{ (S,t) \in \mathbb{R}%
^{+}\times \left[T_{1},T_2\right): S> L_2\right\}\\
&\cD^{L_2}_2=\left\{ (S,t) \in \mathbb{R}%
^{+}\times \left[t^*\vee T_{1},T_2\right]: B^{L,2}(t)=B(t)<S< L_2\right\}
\end{align*}
where $\cC^{L_2}$ is the continuation set and $\cD^{L_2}_1$, $\cD^{L_2}_2$ are two parts of the exercise region.
We also define the local time processes $\ell^{B^{L,2}}=(\ell^{B^{L,2}}_t)_{t\in [T_1,T_2]}$ and $\ell^{B}=(\ell^{B}_t)_{t\in [t^*\vee T_1,T_2]}$ at $B^{L,2}$ on $[T_1,T_2]$
and at $B$ on $[t^*\vee T_1,T_2]$, respectively,
\begin{align*}
 &\ell^{B^{L,2}}_t= Q-\lim_{\varepsilon \downarrow 0}\frac{1}{2\varepsilon}\int_{T_1}^{t}1_{\{B^{L,2}(u)-\varepsilon<S_u<B^{L,2}(u)+\varepsilon\}}d\left \langle S \right \rangle_u\\
 &\ell^{B}_t= Q-\lim_{\varepsilon \downarrow 0}\frac{1}{2\varepsilon}\int_{t^*\vee T_1}^{t}1_{\{B(u)-\varepsilon<S_u<B(u)+\varepsilon\}}d\left \langle S \right \rangle_u.
\end{align*}

Now, we verify that $C^{A,L_2}$ and $B^{L,2}$ satisfy
the conditions of the local time-space formula on curves (Peskir (2005), see Theorem 3.1 \& Remark 3.2).
Indeed, (i) $C^{A,L_2}$ is $C^{1,2}$ in $\cC^{L_2}$ (by strong Markov property) and is smooth in both $\cD^{L_2}_1$ and $\cD^{L_2}_2$; (ii) $B^{L,2}$ and $L_2$ are continuous and of bounded variation as the former is non-increasing and the latter is a constant threshold; (iii) $Q(X_t=B^{L,2}(t))=0$ and $Q(X_t=L_2)=0$ for any $t\in [T_1,T_2]$ as $X$ is GBMP;
(iv) $C^{A,L_2}_t \p\L_{S} C^{A,L_2}\m rC^{A,L_2}$ is locally bounded in $\cC^{L_2}\cup \cD^{L_2}_1\cup \cD^{L_2}_2$ as it is zero in $\cC^{L_2}$,
and as $C^{A,L_2}(S,t)=L_2 - K$  in $\cD^{L_2}_1$ and $C^{A,L_2}(S,t)=S - K$  in $\cD^{L_2}_2$;
(v) $C^{A,L_2}$ is convex on $\cC^{L_2}$, $\cD^{L_2}_1$, and $\cD^{L_2}_2$, respectively;
(vi) $t\mapsto C^{A,L_2}_{S}(t,L_2+)=0$ and $t\mapsto C^{A,L_2}_{S}(t,B^{L,2}(t)-)$ are continuous on $[T_1,T_2]$ (the latter holds due to the $C^{1,2}$-property of $C^{A,L_2}$ in $\cC^{L_2}$);
(vii) $t\mapsto C^{A,L_2}_{S}(t,L_2-)=1$ is continuous on $[t^*\vee T_1,T_2]$.

We can then apply Peskir's formula for $e^{-r(T_2-t)}C^{A,L_2}(S_{T_2},T_2)$, implying that
\begin{align} \label{th-proof} \hs{1pc}
e^{-r(T_2-t)}&C^{A,L_2}(S_{T_2},T_2)\\
=\;&C^{A,L_2}(S,t)+ \int_t^{T_2}
e^{-r(u-t)}\left(C^{A,L_2}_t \p\L_{S} C^{A,L_2}
\m rC^{A,L_2}\right)(S_u,u)du+M_{T_2}\nonumber\\
 &+\frac{1}{2}\int_t^{T_2}
e^{-r(u-t)}\left(C^{A,L_2}_S (S_u+,u)-C^{A,L_2}_S (S_u-,u)\right)1_{\{S_u= B^{L,2}(u)\}}d\ell^{B^{L,2}}_u\nonumber\\
 &+\frac{1}{2}\int_{t^*\vee t}^{T_2}
e^{-r(u-t)}\left(C^{A,L_2}_S (S_u+,u)-C^{A,L_2}_S (S_u-,u)\right)1_{\{S_u=L_2\}}d\ell^{L_2}_u\nonumber\\
=\;&C^{A,L_2}(S,t)+M_{T_2}- r(L_2\m K)\int_t^{T_2}
e^{-r(u-t)} 1_{\{S_u\ge L_2\}}du\nonumber\\
&+\int_{t}^{T_2} e^{-r(u-t)}(rK\m\delta S_u)1_{\{ B(u)\le S_u\le L_2\}}du\nonumber\\
 &+\frac{1}{2}\int_t^{t^* \vee t}
e^{-r(u-t)}\left(C^{A,L_2}_S (S_u+,u)-C^{A,L_2}_S (S_u-,u)\right)1_{\{S_u= L_2\}}d\ell^{L_2}_u\nonumber\\
 &+\frac{1}{2}\int_{t^*\vee t}^{T_2}
e^{-r(u-t)}\left(C^{A,L_2}_S (S_u+,u)-C^{A,L_2}_S (S_u-,u)\right)1_{\{S_u= B(u)\}}d\ell^{B}_u\nonumber\\
 &+\frac{1}{2}\int_{t^*\vee t}^{T_2}
e^{-r(u-t)}\left(C^{A,L_2}_S (S_u+,u)-C^{A,L_2}_S (S_u-,u)\right)1_{\{S_u=L_2\}}d\ell^{L_2}_u\nonumber\\
=\;&C^{A,L_2}(S,t)+M_{T_2}- r(L_2\m K)\int_t^{T_2}
e^{-r(u-t)} 1_{\{S_u\ge L_2\}}du\nonumber\\
&+\int_{t}^{T_2} e^{-r(u-t)}(rK\m\delta S_u)1_{\{ B(u)\le S_u\le L_2\}}du
-\frac{1}{2}\int_t^{T_2}
e^{-r(u-t)}C^{A,L_2}_S (L_2-,u)d\ell^{L_2}_u\nonumber
  \end{align}
where $M=(M_t)_{t\ge T_1}$ is the martingale part
and we used that $C^{A,L_2}_t \p\L_{S} C^{A,L_2}
\m rC^{A,L_2}=0$ in $\cC^{L_2}$; $C^{A,L_2}(S,t)=L_2 - K$  in $\cD^{L_2}_1$ and $C^{A,L_2}(S,t)=S - K$  in $\cD^{L_2}_2$; splitting the local time term w.r.t $\ell^{B^{L,2}}$ into two parts; the smooth-fit condition at $B$ on $[t^*\vee T_1,T_2]$;
and finally that $C^{A,L_2}_S(L_2+,\cdot)=0$.
Now
upon taking the expectation $\EE_t$, using the optional sampling theorem for $M$, rearranging terms and noticing that
$C^{A,L_2}(S,T_2)=(S\wedge L_2- K)^+$ for all $S>0$, we get \eqref{capped-price-3}.
\end{proof}

\begin{figure}[t]
\begin{center}
\includegraphics[scale=0.75]{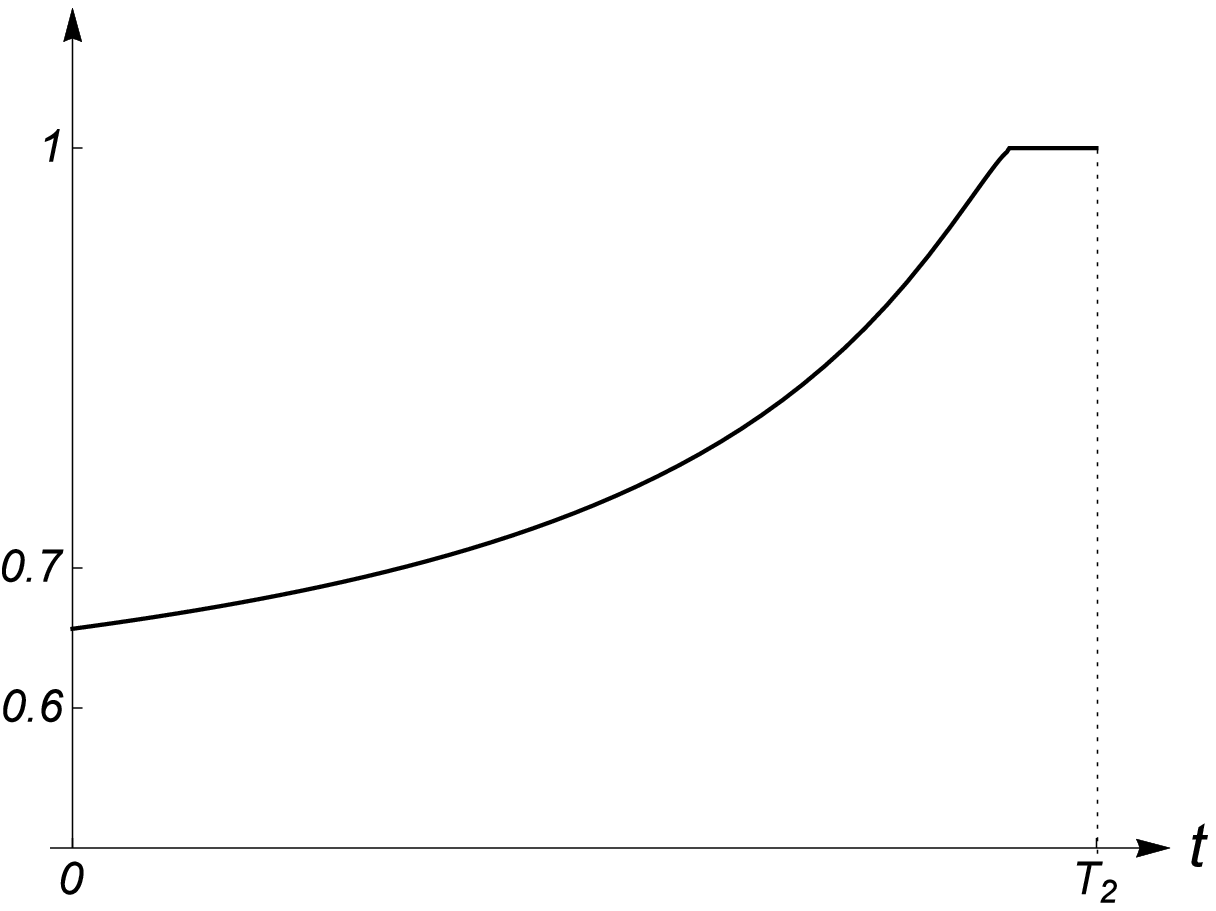}
\end{center}
\par
{}
\par
\leftskip=1.6cm \rightskip=1.6cm {\small \noindent \vspace{-10pt} }
\par
{\small \textbf{Figure 3.} This figure plots the left-hand side derivative $%
C_{S}^{A,L_{2}}(L_{2}-,t)$ of $C^{A,L_{2}}$ at $L_2$ for $t\in [0,T_2)$. The
parameter set is $T_2=4, K=1, L_2=1.39, r=0.1, \delta=0.1, \sigma=0.3$. }
\par
\vspace{10pt}
\end{figure}

\vspace{6pt}


\begin{remark}
Using the definition of the local time and by straightforward calculations
one can obtain that
\begin{equation}
\hspace{3.0112pc}d\mathsf{E}_{t}\left[ \ell _{u}^{L_{2}}\right] =\varphi
\left( -\tfrac{\log {(L_{2}/S)}-(r-\delta -\frac{\sigma ^{2}}{2})(u-t)}{%
\sigma \sqrt{u-t}}\right) \frac{\sigma L_{2}}{\sqrt{u-t}}du
\label{expected-local-time-2}
\end{equation}%
where $\varphi (z)=\frac{1}{\sqrt{2\pi }}e^{-\frac{z^{2}}{2}}$ is the
probability density function of the standard normal law, and thus the last
term in \eqref{capped-price-3} can be rewritten by inserting the expectation
inside so that it becomes a Lebesgue-Stieltjes integral with respect to the
expected local time
\begin{equation*}
\hspace{3.0112pc}\frac{1}{2}%
\int_{t}^{T_{2}}e^{-r(u-t)}C_{S}^{A,L_{2}}(L_{2}-,u)\;\varphi \left( -\tfrac{%
\log {(L_{2}/S)}-(r-\delta -\frac{\sigma ^{2}}{2})(u-t)}{\sigma \sqrt{u-t}}%
\right) \frac{\sigma L_{2}}{\sqrt{u-t}}du.
\end{equation*}
\end{remark}

\begin{remark}
We note that formula \eqref{capped-price-3} for $C^{A,L_2}$ is implicit and
requires the knowledge of $C_{S}^{A,L_{2}}(L_{2}-,u)$, which can be
estimated from the formula \eqref{capped-price-1} for $u\in [T_{1},t^{\ast
}) $ and is given by $C_{S}^{A,L_{2}}(L_{2}-,u)=1$ for $u\in [t^{\ast
},T_{2})$ (see Figure 3). This way of computing $C^{A,L_{2}}$ was found to
be quite efficient as it does not use formula \eqref{capped-price-1}
globally, but only below $L_{2}$, to estimate the left derivative $%
C_{S}^{A,L_{2}}(L_{2}-,u)$.
\end{remark}

\begin{remark}
If we take $t\geq \lbrack t^{\ast }\vee T_{1},T_{2}]$ and $S=B^{L,2}(t)$ in %
\eqref{capped-price-3}, then straightforward manipulations show that $%
B^{L,2} $ solves the following integral equation
\begin{equation*}
\hspace{3.0112pc}B^{L,2}(t)-K=C^{E}(B^{L,2}(t),t)+\mathsf{E}_{t}\left[
\int_{t}^{T_{2}}e^{-r(u-t)}(\delta S_{u}\!-\!rK)1_{\{S_{u}\geq
B^{L,2}(u)\}}du\right]
\end{equation*}%
on $[t^{\ast }\vee T_{1},T_{2}]$ where $C^{E}$ is the price of the
corresponding European uncapped call option. We note that this integral
equation is exactly the same as the integral equation for the American
uncapped call option boundary. This is not surprising as we know from
Theorem 2.1 that $B^{L,2}=B$ on $[t^{\ast }\vee T_{1},T_{2}]$. Therefore,
Theorem 2.3 provides alternative proof of this fact which was proven in
Broadie and Detemple (1995).
\end{remark}

In the next section we turn to the interval $\left[ 0,T_{1}\right] $.

\section{Main Result}

This section contains the main result of the paper, Theorem \ref{th:case1},
which describes the optimal exercise policy on the interval $\left[ 0,T_{1}%
\right] $ for all cases of interest. The complete details, proofs and the
option pricing formulas are given in Sections 4-6.

\begin{theorem}
\label{th:case1} (i) If $L_{1}<\min (L_{2},B(T_{1}))$, then the optimal
exercise time in \eqref{problem-1} is given by
\begin{equation}
\hspace{4.015pc}\tau ^{\ast }(t)=\inf \left\{ u\geq t:(S_{u},u)\in \mathcal{E%
}\right\} \wedge T_{2}  \label{tau-star-1}
\end{equation}%
for $t\in \lbrack 0,T_{1})$, where the exercise region $\mathcal{E}=\mathcal{%
E}_{1}\cup \mathcal{E}_{2}$ consists of two disjoint regions
\begin{align}
\hspace{4.015pc}& \mathcal{E}_{1}=\left\{ \left( S,t\right) \in \mathbb{R}%
^{+}\times \left[ 0,t^{1}\right] :L_{1}\leq S\leq B^{L,1}(t)\right\} \\
& \mathcal{E}_{2}=\left\{ \left( S,t\right) \in \mathbb{R}^{+}\times \left[
T_{1},T_{2}\right] :S\geq B^{L,2}(t)\right\}
\end{align}%
for some $t^{1}\in \lbrack 0,T_{1})$ (specified in the next section). The
exercise boundary $B^{L_{1}}=(B^{L,1}(t))_{t\in \left[ 0,t^{1}\right] }$ is
characterized as
\begin{equation*}
\hspace{4.015pc}B^{L,1}(t)=L_{1}\qquad \text{for}\;t\in \lbrack T_{0},t^{1}]
\end{equation*}%
where $0\leq T_{0}\le t^{1}$ is described in the next section as well, and
on the interval $[0,T_{0}]$ the boundary $B^{L,1}$ solves the integral
equation
\begin{equation}
\hspace{4.015pc}L_{1}-K=C^{E,L}(B^{L,1}(t),t)+\Pi
(B^{L,1}(t),t;B^{L,1}(\,\cdot \,))
\end{equation}%
subject to the boundary condition $B^{L,1}(T_{0}-)=L_{1}$, where the
functions $C^{E,L}(S,t)$ and $\Pi (S,t;B^{L,1}(\,\cdot \,))$ are given in %
\eqref{C-EL} and \eqref{Pi}, respectively. Moreover, if
\begin{equation}
\hspace{4.015pc}t^{0}\equiv T_{1}-\frac{1}{r}\log ((L_{2}-K)/(L_{1}-K))\geq 0
\end{equation}%
then $B^{L,1}(t)=+\infty $ for $t\in \lbrack 0,t^{0}]$ and $%
B^{L,1}(t)<+\infty $ for $t\in(t^0,t^1].$ \vspace{6pt}

(ii) If $B(T_{1})\leq L_{1}<L_{2}$, then the optimal exercise time in %
\eqref{problem-1} is given by
\begin{equation}
\hspace{4.0452pc}\tau ^{\ast }(t)=\inf \left\{ u\geq t:(S_{u},u)\in \mathcal{%
E}\right\} \wedge T_{2}  \label{tau-star-1}
\end{equation}%
for $t\in \lbrack 0,T_{1})$, where the exercise region $\mathcal{E}=\mathcal{%
E}_{1}\cup \mathcal{E}_{2}$ consists of
\begin{align}
\hspace{4.0452pc}\mathcal{E}_{1}=& \left\{ \left( S,t\right) \in \mathbb{R}%
^{+}\times \left[ 0,T_{1}\right) :L_{1}\leq S\leq B^{L,1}(t)\right\}  \\
& \cup \left\{ \left( S,t\right) \in \mathbb{R}^{+}\times \left[ t^{\ast
},T_{1}\right) :B(t)\leq S\leq L_{1}\right\}   \notag \\
\mathcal{E}_{2}=& \left\{ \left( S,t\right) \in \mathbb{R}^{+}\times \left[
T_{1},T_{2}\right] :S\geq B(t)\right\} .
\end{align}%
The exercise boundary $B^{L_{1}}=(B^{L,1}(t))_{t\in \lbrack 0,T_{1})}$ is
characterized as
\begin{equation*}
\hspace{4.0452pc}B^{L,1}(t)=L_{1}\qquad \text{for}\;t\in \lbrack T_{0},T_{1}]
\end{equation*}%
where $0\leq T_{0}<T_{1}$ (specified in Section 5), and on the interval $%
[0,T_{0}]$ the boundary $B^{L,1}$ solves the integral equation
\begin{equation*}
\hspace{4.0452pc}L_{1}-K=C^{E,L}(B^{L,1}(t),t)+\Pi
(B^{L,1}(t),t;B^{L,1}(\,\cdot \,))
\end{equation*}%
subject to the boundary condition $B^{L,1}(T_{0}-)=L_{1}$, where $%
C^{E,L}(S,t)$ and $\Pi (S,t;B^{L,1}(\,\cdot \,))$ are given in %
\eqref{C-EL-case2} and \eqref{Pi-case2}, respectively. Moreover, if $%
t^{0}\equiv T_{1}-\frac{1}{r}\log ((L_{2}-K)/(L_{1}-K))\geq 0$, then $%
B^{L,1}(t)=+\infty $ for $t\in \lbrack 0,t^{0}]$ and $B^{L,1}(t)<+\infty $
for $t\in (t^{0},T_{1}]$. \vspace{6pt}

(iii) If $L_1>L_2$ and the cap is left-continuous, then the optimal exercise
time in \eqref{problem-1} is given by
\begin{equation}  \label{tau-star-1}
\hspace{4pc}\tau ^{\ast }(t)=\inf \left\{ u\geq t:(S_{u},u)\in \mathcal{E}%
\right\} \wedge T_{2}
\end{equation}
for $t\in[0,T_1)$, where the exercise region is
\begin{equation}
\mathcal{E}=\left\{ \left( S,t\right) \in \mathbb{R}^{+}\times \left[ 0,T_{1}%
\right] :S\geq B^{L,1}\left( t\right) \right\} \cup \left\{ \left(
S,t\right) \in \mathbb{R}^{+}\times \left( T_{1},T_{2}\right] :S\geq
B(t)\wedge L_{2}\right\}.
\end{equation}%
The optimal exercise boundary $B^{L,1}=(B^{L,1}(t))_{t\in[0,T_1]}$ is decreasing and continuous
on $[0,T_1)$ and is characterized as the solution to the integral equation
\begin{align}
\hspace{4pc} B^{L,1}(t)\wedge L_1-K=C^{E,L}(
B^{L,1}(t),t)+\Pi(B^{L,1}(t),t;B^{L,1}(\,\cdot\,))
\end{align}
for $t\in[0,T_1)$ subject to the boundary condition $B^{L,1}(T_1-)=\max(rK/%
\delta,L_2\wedge B(T_1))$, where the functions $C^{E,L}(S,t)$ and $\Pi
(S,t;B^{L,1}(\,\cdot \,))$ are given in \eqref{C-EL-case3} and %
\eqref{Pi-case3}, respectively. The value of the boundary at $T_1$ is $%
B^{L,1}(T_1)=L_2\wedge B(T_1)$.
\end{theorem}

Theorem \ref{th:case1} shows that the structure of the optimal exercise
region over $\left[ 0,T_{1}\right] $ depends on the relative positions of
the cap components $L_{1},L_{2}$ and of the boundary $B\left( T_{1}\right) $
of the uncapped option at $T_{1}$. The first case is when $L_{1}<L_{2}\wedge
B\left( T_{1}\right) $. In this instance, the exercise region involves the
times $0\leq t^{0}\vee 0\le T_{0}\le t^{1}<T_{1}$. In the interval $\left(
t^{1},T_{1}\right) $, immediate exercise is always suboptimal: $\left(
S,t\right) \notin \mathcal{E}$ for all $S\in \mathbb{R}^{+}$ and $t\in
(t^{1},T_{1})$. In $\left[ T_{0},t^{1}\right] $, immediate exercise is
optimal along the line $S=L_{1}$: for $t\in \left[ T_{0},t^{1}\right] $, $%
\left( S,t\right) \in \mathcal{E}_{1}$ if and only if $S=L_{1}$. In $\left(
t^{0},T_{0}\right) $, it is optimal in the region between the boundaries $%
L_{1}$ and $B^{L,1}=\left( B^{L,1}\left( t\right) \right) _{t\in \left(
t^{0},T_{0}\right) }$: for $t\in \left( t^{0},T_{0}\right) $, $\left(
S,t\right) \in \mathcal{E}_{1}$ if and only if $L_{1}\leq S\leq
B^{L,1}\left( t\right) $. Finally, in $[0,t^{0})$, it is optimal to exercise
immediately above the cap $L_{1}$: for $t\in \lbrack 0,t^{0})$, $\left(
S,t\right) \in \mathcal{E}_{1}$ if and only if $S\geq L_{1}$. The most
noteworthy aspect of the exercise region is that it splits into two parts $%
\mathcal{E}_{1}$ and $\mathcal{E}_{2}$, separated by a region, $\left\{
\left( S,t\right) :S\in \mathbb{R}^{+},t\in (t^{1},T_{1})\right\} $, in
which it is optimal to continue. Figure 6 displays the exercise region $%
\mathcal{E}=\mathcal{E}_{1}\cup \mathcal{E}_{2}$ and the times $\left(
t^{0},T_{0},t^{1}\right) $ for this case.

The second case is when $B(T_{1})\leq L_{1}<L_{2}$. In this instance $%
t^{1}=T_{1}$ and the exercise region involves the times $0\leq t^{0}\vee 0\leq
T_{0}< T_{1}$. Thus, the continuation subregion separating the sets $%
\mathcal{E}_{1}$ and $\mathcal{E}_{2}$ has disappeared. But the exercise
region $\mathcal{E}_{1}$ now includes the set $\left\{ \left( S,t\right) \in
\mathbb{R}^{+}\times \left[ t^{\ast },T_{1}\right) :B(t)\leq S\leq
L_{1}\right\} $ below the cap $L_{1}$. Figure 9 illustrates this structure
of $\mathcal{E}=$ $\mathcal{E}_{1}\cup \mathcal{E}_{2}$.

The last case is when $L_{1}>L_{2}$ and the cap is left-continuous, i.e., at
time $T_{1}$ the cap is $L(T_{1})=L_{1}$. In this instance immediate
exercise is always optimal above the cap $L_{1}$ for $t\in \left[ 0,T_{1}%
\right] $. Moreover, there exists a lower boundary $B^{L}=(B^{L}(t))_{t\in
\lbrack 0,T_{1}]}$ such that it is optimal to exercise at $\left( S,t\right)
$ if $B^{L}(t)\leq S\leq L_{1}$ and $t\in \lbrack 0,T_{1}]$. Figure 10 shows
$\mathcal{E}$ for this case.

Intuition for the various components of the exercise region can be provided
as follows. The time $t^{0}$, which appears in cases (i) and (ii), is the
unique solution to the equation $e^{-r\left( T_{1}-t^{0}\right) }\left(
L_{2}-K\right) =L_{1}-K$. If $t^0\ge 0$, then the left-hand side of this equation is the value
at $t^{0}\in \left[ 0,T_{1}\right] $ from receiving the maximum payoff on
the contract at the earliest possible date, $T_{1}$. At times $t$ prior to $%
t^{0} $, immediate exercise above the cap $L_{1}$ pays off $%
L_{1}-K>e^{-r\left( T_{1}-t\right) }\left( L_{2}-K\right) $, hence cannot be
dominated by any waiting policy.

The time $t^{1}$, which plays a major role in case (i), \ is determined by
the policy of waiting until the jump time $T_{1}$ before deciding to
exercise or not. The value of this particular waiting policy is the value $%
C^{w}\left( S,t^{1}\right) $ of a European compound option with maturity
date $T_{1}$ paying the American capped call option price at time $T_{1}$.
As $L_{1}<\min (L_{2},B(T_{1}))$, the payoff at $\left( L_{1},T_{1}\right) $
is $C^{A,L_{2}}\left( L_{1},T_{1}\right) >L_{1}-K$. Hence, immediate
exercise is suboptimal at the point $\left( L_{1},T_{1}\right) $. Moving
back in time along the cap $L_{1}$, continuously changes the value $%
C^{w}\left( L_{1},t\right) $ of the compound contract, because of
discounting and uncertainty in the underlying asset price, but the immediate
exercise payoff remains the same. The time $t^{1}$ is the largest time in $%
\left[ t^{0},T_{1}\right) $ at which $C^{w}\left( L_{1},t\right) =L_{1}-K$.
For any $t\in \left( t^{1},T_{1}\right) $, $C^{w}\left( L_{1},t\right)
>L_{1}-K$ and immediate exercise at $\left( L_{1},t\right) $ is dominated.
Moreover, because $C^{w}\left( S,t\right) $ is an increasing function of $S$%
, $C^{w}\left( S,t\right) >L_{1}-K$ and immediate exercise at $\left(
S,t\right) $ is also dominated for $S>L_{1}$ and $t\in \left(
t^{1},T_{1}\right) $. That immediate exercise is dominated for $S<L_{1}$ and
$t\in \left( t^{1},T_{1}\right) $ follows from the fact that a capped option
with cap $L_{1}$ and maturity date $T_{1}$ would not be exercised at this
point.

Time $t^{1}$ is the largest time in $\left[ t^{0},T_{1}\right) $ at which it
becomes optimal to exercise along the cap. As $C^{w}\left( \cdot
,t^{t}\right) $ is strictly increasing and the exercise payoff of the capped
option does not change above the cap, $C^{w}\left( S,t^{1}\right) >L_{1}-K$
for $S>L_{1}$ and immediate exercise remains suboptimal at $\left(
S,t^{1}\right) $ for $S>L_{1}$. It is also suboptimal at $\left(
S,t^{1}\right) $ for $S<L_{1}$ for the reasons indicated above. This
knife-edge property at $S=L_{1}$ persists as one moves further back in time
along the cap. The time $T_{0}$ marks the end of that stretch. For $t<T_{0}$%
, immediate exercise can become optimal above the cap, specifically for $%
L_{1}\leq S\leq B^{L,1}(t)$ where $B^{L,1}=\left( B^{L,1}\left( t\right)
\right) _{t\in \left( t^{0},T_{0}\right) }$ is an endogenous boundary.

Intuition for case (iii) is as follows. As we impose left-continuity of the
cap, we know the optimal exercise strategy on $(T_{1},T_{2}]$. At time $%
t=T_{1}$, we compare the immediate payoff $S\wedge L_{1}-K$ with the
continuation value, i.e., the price of the capped option $%
C^{A,L_{2}}(S,T_{1})$ (with cap $L_{2}$). Clearly, immediate exercise is
optimal if and only if $S$ is above the critical threshold $%
b^{L,1}(T_{1})=L_{2}\wedge B(T_{1})$. It then remains to understand the
exercise policy on $[0,T_{1})$. Immediate exercise above $L_{1}$ is always
optimal because $L_{1}-K$ is the maximum payoff that can be attained. As the
cap is non-increasing in this case, we have that there exists a boundary $%
B^{L,1}$ on $[0,T_{1})$ below the cap $L_{1}$ such that it is optimal to
exercise if $B^{L,1}(t)\leq S\leq L_{1}$. Furthermore, one should not
exercise if $S<rK/\delta \wedge L_{1}$, because the local benefits of
waiting to exercise are positive in this region. Near $T_{1}$, it is optimal
to exercise the option when $S\geq \max (rK/\delta ,L_{2}\wedge B(T_{1}))$
as these benefits are negative. Therefore, there is a part of the exercise
region below $L_{1}$ and prior to $T_{1}$ if and only if $rK/\delta <L_{1}$.


\begin{remark}
(Non dividend-paying asset) If $\delta =0$, then it is never optimal to
exercise the American uncapped call option early, i.e., $B=+\infty $ on $%
[0,T_{2}]$. This immediately excludes the second case. The exercise region
on $[T_{1},T_{2}]$ is simply the set above $L_{2}$. The structure of $%
\mathcal{E}_{1}$ on $[0,T_{1})$, in the first and third cases, has the form
described above. The condition $\delta =0$ does not simplify it. To
summarize, the only qualitative change in this case is that the exercise
boundary on $[T_{1},T_{2}]$ becomes the constant cap $L_{2}$.
\end{remark}

\begin{remark}
(Constant cap) If $L_{1}=L_{2}$, then the contract simply becomes the
American capped option with single cap $L=L_{1}=L_{2}$ from Broadie and
Detemple (1995) (see the results of Section 2).
\end{remark}

\vspace{6pt}

\section{Case $L_1<\min (L_2,B(T_1))$}

In this section we assume that $L_{1}<B(T_{1})$ or, equivalently, that $%
t^{\ast }\geq T_{1}$, and the goal is to establish the statements of Theorem %
\ref{th:case1}$(i)$ above. The complete analysis of this case is divided
into steps 1-8 below.

\begin{proof}[Proof of Theorem \protect\ref{th:case1}$(i)$]
\vs{2pt}

1. Let us define the function%
\begin{equation*}
C^{w}\left( S,t\right) \equiv \mathsf{E}_{t}\left[ e^{-r\left(
T_{1}-t\right) }C^{A,L_{2}}\left( S_{T_{1}},T_{1}\right) \right] =\mathsf{E}%
_{t}\left[ e^{-r\left( \tau ^{*}(T_1)-t\right) }\left( S_{\tau ^{*}(T_1)}\wedge
L_{2}-K\right) ^{+}\right]
\end{equation*}%
for $t<T_{1}$ and $S>0$, where $\tau^{*}(T_1)$ is defined in \eqref{tau-star} and in the second
equality we used Theorem 2.1. The function $C^{w}\left( S,t\right) $ is the
value of a European derivative with maturity date $T_{1}$ and payoff given
by the American capped call option price with cap $L_{2}$. It can also be
viewed as an American capped call option with delayed exercise provision
restricted to $\left[ T_{1},T_{2}\right] $. For the holder of the original
capped option \eqref{problem-1}, $C^{w}\left( S,t\right) $ is the value of
the policy which consists in waiting until at least $T_{1}$ before deciding
to exercise or not. This waiting policy is always feasible and its value
provides a lower bound for $C^{A,L}$.
By using the
Ito-Tanaka formula for $C^{A,L_2}$ we can rewrite $C^w$ as
\begin{align}\label{Cw}
C^{w}(S,t)=& \mathsf{E}_{t}\left[ e^{-r(T_{1}-t)}C^{A,L_{2}}(S_{T_{1}},T_{1})%
\right]  \\
=& C^{A,L_{2}}(S,t)-r(L_{2}-K)\mathsf{E}_{t}\left[\int_{t}^{T_{1}}e^{-r(u-t)}
1_{\{S_{u}\geq L_{2})\}} du \right] \notag \\
& -\frac{1}{2}\mathsf{E}%
_{t}\left[\int_{t}^{T_{1}}e^{-r(u-t)}C_{S}^{A,L_{2}}(L_{2}-,u)d\ell _{u}^{L_{2}}\right]  \notag\\
=& C^{E,L_{2}}(S,t)+r(L_{2}-K)\mathsf{E}_{t}\left[\int_{T_1}^{T_{2}}e^{-r(u-t)}
1_{\{S_{u}\geq L_{2})\}} du \right] \notag \\
&+\mathsf{E}_{t}\left[\int_{t}^{T_2} e^{-r(u-t)}(\delta S_u\m rK)1_{\{
B(u)\le S_u\le L_2\}}du \right]  \notag \\
& +\frac{1}{2}\mathsf{E}%
_{t}\left[\int_{T_1}^{T_{2}}e^{-r(u-t)}C_{S}^{A,L_{2}}(L_{2}-,u)d\ell _{u}^{L_{2}}\right]  \notag
\end{align}%
for $S>0$ and $t<T_1$. We then use \eqref{expected-local-time-2} to calculate $d\EE_t[\ell^{L_2}_u]$, the derivative
$C^{A,L_2}_S(L_2-,u)$ is estimated using \eqref{capped-price-1}, and $\mathsf{E}_{t}[1_{\{S_{u}\geq L_{2}\}}]$,
$\mathsf{E}_{t}[1_{\{B(u)\le S_u\le L_2\}}]$,
$\mathsf{E}_{t}[S_u 1_{\{B(u)\le S_u\le L_2\}}]$
can be written in terms of the cumulative distribution function $\Phi(\;\cdot\;)$ of the standard normal law.
\vspace{2pt}

As $L_{1}<B(T_{1})$, we have that $C^{w}\left( L_{1},T_{1}\right)
>L_{1}-K $. Then the unique solution $B^{w}(T_1)$ to
\begin{equation}
\hspace{4pc} C^{w}\left( B^{w}(T_1),T_{1}\right) =L_{1}-K
\end{equation}
is such that $B^{w}(T_1)<L_{1}$. Similarly, for $t\in \left[ 0,T_{1}\right)$%
, let $B^{w}(t)$ be the unique solution to%
\begin{equation}  \label{B1}
\hspace{4pc} C^{w}\left( B^{w}(t),t\right) =L_{1}-K
\end{equation}
and note that%
\begin{equation*}
\hspace{4pc} \left\{
\begin{array}{cc}
C^{w}\left( S,t\right) >L_{1}-K & \text{for }S>B^{w}(t) \\
C^{w}\left( S,t\right) <L_{1}-K & \text{for }S<B^{w}(t).%
\end{array}%
\right.
\end{equation*}%
The boundary $B^{w}=(B^{w}(t))_{t\in[0,T_1]}$ represents the continuous curve at which the waiting
policy associated with $C^w$ described above, has the same value as $L_1-K$, that is the largest possible payoff on $[0,T_1)$. The
following result is an obvious consequence.

\begin{lemma}
\label{ls} Immediate exercise of the capped option is suboptimal if $t\in %
\left[ 0,T_{1}\right] $ and $S>B^{w}(t)$ or $S<L_{1}$.
\end{lemma}

\begin{proof}
If $S>B^{w}(t)$, then $C^{w}( S,t) >L_{1}-K\ge S\wedge L_1-K$ implying that
immediate exercise is strictly dominated by the policy of waiting until time $T_{1}$
before considering exercise. If $S<L_{1}$, then there exists an American
uncapped call option $C^{A}\left( S,t;T_{0}\right)$ with shorter maturity $T_{0}\leq T_{1}$ and corresponding
optimal exercise boundary $B^0$ such that $S<B^{0}(t)<L_1$
and thus $C^{A}\left( S,t;T_{0}\right) >S-K$. As the policy $B^{0}$ is feasible for the capped option holder and the payoff matches the
payoff of the uncapped option, it follows that $C^{A,L}\left( S,t\right)
\geq C^{A}\left( S,t;T_{0}\right) >S-K$. Immediate exercise is therefore
suboptimal for the original American capped option.
\end{proof}
\vspace{6pt}

2. Let us define $t^{1}$ as
the largest root of $B^{w}(t)=L_{1}$ on $[0,T_1]$ if one exists, otherwise
we let $t^{1}=0$,
and note that $t^{1}<T_{1}$ as $B^{w}(T_{1})<L_{1}$. Also define the set
\begin{equation}
\hspace{4pc} \mathcal{D}^w=\left\{ \left( S,t\right) \in \mathbb{R}%
^{+}\times \left[ 0,t^{1}\right] :L_{1}\leq S\leq \max(B^{w}(t),L_1) \right\}
\label{Dw}
\end{equation}
which represents the region with lower boundary $L_{1}$ and upper boundary $%
\max(B^{w},L_1)$ on $[0,t^1]$ (see Figure 4 for illustration). Now we are ready to provide
further insights into the structure of the optimal exercise region.

\begin{figure}[t]
\begin{center}
\includegraphics[scale=0.75]{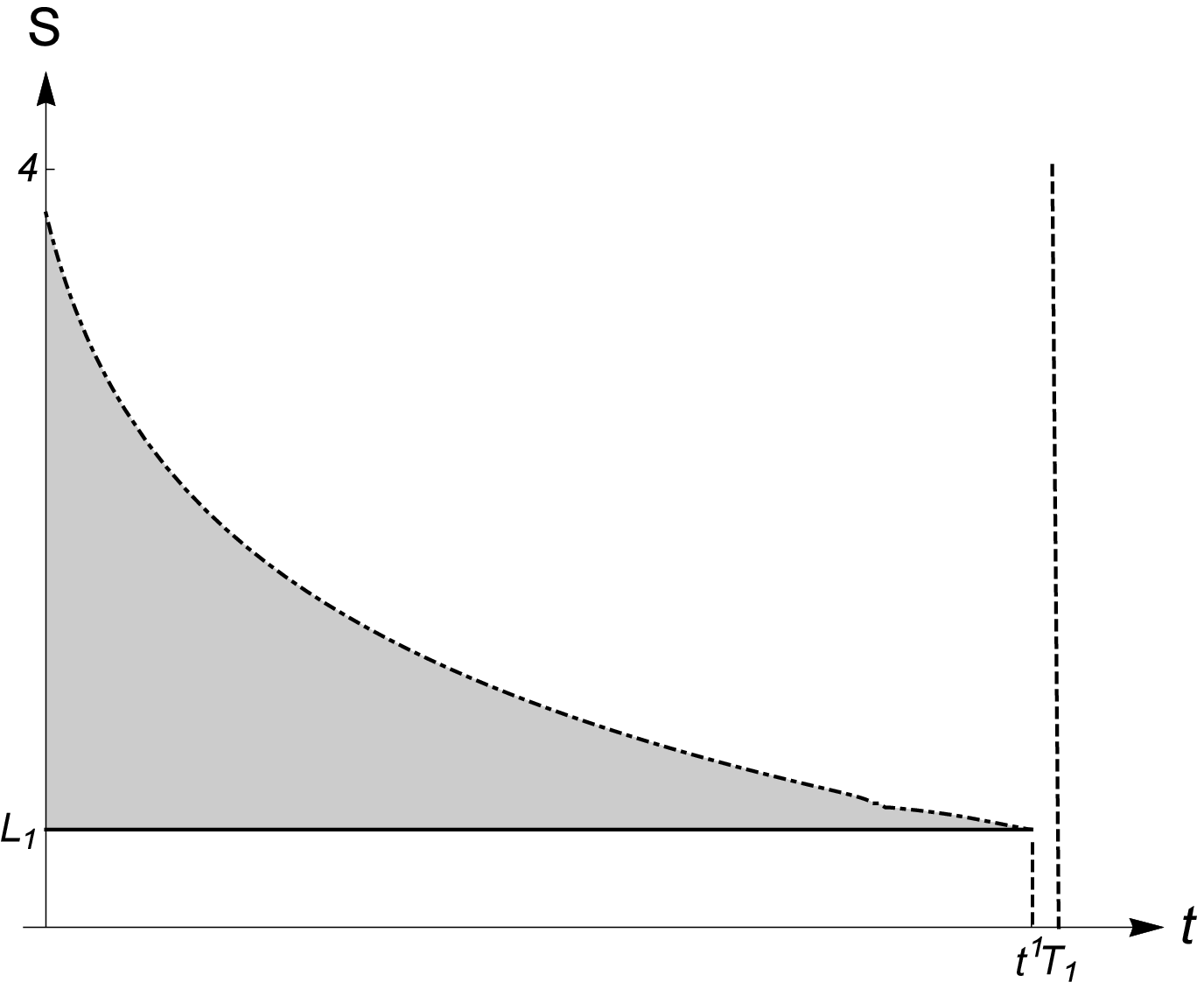}
\end{center}

{\par \leftskip=1.6cm \rightskip=1.6cm \small \ni \vs{-10pt}

\textbf{Figure 4.} This figure plots the region $\mathcal{D}^w$ (gray). The dot-dashed line
represents the upper boundary $B^w$.
 The parameter set is $T_1=3, T_2=4, K=1, L_1=1.3, L_2=1.39, r=0.1, \delta=0.1, \sigma=0.3$. The value of $t^1$ is 2.93.

\par} \vs{10pt}

\end{figure}

\begin{theorem}
\label{td} $(i)$ The immediate exercise region of the capped call option %
\eqref{problem-1} consists of two disconnected subregions $\mathcal{E}=%
\mathcal{E}_{1}\cup \mathcal{E}_{2}$ where $\mathcal{E}_{1}\subseteq
\mathcal{D}^w$ and
\begin{equation*}
\hspace{4pc} \mathcal{E}_{2}=\left\{ \left( S,t\right) \in \mathbb{R}%
^{+}\times \left[ T_{1},T_2\right] :S\geq B^{L,2}(t)\right\}.
\end{equation*}

$(ii)$ The subregion $\mathcal{E}_{1}$ is down-connected above $L_1$ and
non-empty if $t^{1}>0$. If $t^{1}=0$ and $B^{w}(0)=L_{1}$, it consists of
the single point $\left( L_{1},t^{1}\right) $. If $t^{1}=0$ and $%
B^{w}(0)<L_{1}$, then $\mathcal{E}_{1}=\varnothing $.

$(iii)$ $\left( L_{1},v\right) \in \mathcal{E}_{1}$ for all $v\in \left[
0,t^{1}\right]$.

$(iv)$ If $t^0\equiv T_1 -\frac{1}{r}\log((L_2-K)/(L_1-K))\ge 0$, then  $(S,t)\in%
\mathcal{E}_{1}$ for any $t\in[0,t^0]$ and $S\ge L_1$. Also, for any $%
t\in(t_0,T_1)$ there exists $\widetilde{s}=\widetilde{s}(t)$ such that for
any $S>\widetilde{s}$ it is optimal to continue at $(S,t)$. 

$(v)$ There exists $T_0\in[0,t^1]$ such that it is not optimal to exercise
at $(S,t)$ for $S>L_1$ and $t\in[T_0,t^1]$.
\end{theorem}

\begin{proof}
$(i)$ By Lemma \ref{ls}, immediate exercise is suboptimal above $B^{w}$ and below $%
L_{1}$. It follows immediately that $\mathcal{E}_{1}\subseteq \mathcal{D}^w$.

$(ii )$ To show that $\mathcal{E}_{1}$ is down-connected above $L_1$, suppose that $\left(
S,t\right) \in \mathcal{E}_{1}$ and consider a point $\left( S^{\prime
},t\right) $ such that $S\geq S^{\prime }\geq L_{1}$. As $%
S_{\cdot}\geq S_{\cdot}^{\prime }$ it follows that $C^{A,L}\left( S,t\right) \geq
C^{A,L}\left( S^{\prime },t\right)$. Optimality of immediate
exercise at $\left( S,t\right) $ implies $L_{1}-K=C^{A,L}\left( S,t\right)
\geq C^{A,L}\left( S^{\prime },t\right) \geq L_{1}-K$. Thus $\left(
S^{\prime },t\right) \in \mathcal{E}_{1}$, which establishes
down-connectedness. To show non-emptiness, let $t^{1}>0$ and note that
$C^{w}\left( L_{1},t^1\right) =L_{1}-K$. In $%
\left( t^{1},T_{1}\right] $, immediate exercise is suboptimal, so the value
of the American capped option is the same as the value of the option with
exercise delayed to the interval $\left[ T_{1},T_2\right] $, i.e., $%
C^{A,L}\left( S,t\right) =C^{w}\left( S,t\right) $. If $t=t^1$ and
$S=L_{1}$, then there are only two choices, exercise or follow the delayed exercise
policy. Both produce the same value. Thus, $C^{A,L}\left( L_{1},t^{1}\right)
=$ $C^{w}\left( L_{1},t^{1}\right) =L_{1}-K$ and $\left(
L_{1},t^{1}\right) \in \mathcal{E}_{1}$.  The same argument applies if $%
t^{1}=0$ and $B^w(0)=L_{1}$. In this case, the region consists of a single
point, $\mathcal{D}=\left\{ \left( L_{1},t^{1}\right) \right\} $ and $%
\mathcal{E}_{1}=\mathcal{D}$. Finally, if $t^{1}=0$ and $B^w(0)<L_{1}$,
immediate exercise is suboptimal for all $S\in \mathbb{R}^{+}$, implying $%
\mathcal{E}_{1}=\varnothing $.
\vs{2pt}

$(iii)$ From $(i)$ we have that $(L_1,t^1)\in \mathcal{E}_{1}$ and let us take $\left( L_{1},t\right)$ for any
$t\in \left[ 0,t^{1}\right]$. Now, we introduce an auxiliary capped option $C^{A,\wt{L}}$ with earlier maturity $\wt{T}_1<T_1$ for the cap $L_1$ so that
the new cap is given as $\wt{L}_{t}=L_{1}1_{t<\wt{T}_{1}}+L_{2}1_{\wt{T}_{1}\leq t\leq T_2}$. Clearly $C^{A,L}\le C^{A,\wt{L}}$, then we can decrease $\wt{T}_1$ continuously such that the point $\wt{t}^1<\wt{T_1}$ (which is defined in the same way as $t^1$ for $T_1$) coincides with $t$. From $(i$) we then have that
$C^{A,\wt{L}}\left( L_{1},\wt{t}^{1}\right) =L_{1}-K$ and thus $L_1-K\le C^{A,L}\left( L_{1},t\right)\le C^{A,\wt{L}}\left( L_{1},\wt{t}^1\right)  = L_{1}-K$, i.e.,
 $\left( L_{1},t\right) \in \mathcal{E}_{1}$.
\vs{2pt}

$(iv)$ Now if we assume that $t^0=T_1- \frac{1}{r}\log((L_2-K)/(L_1-K))\ge 0$ and take $t\le t^0$. Then
for any exercise policy $\tau$ we have that
\begin{align} \hs{3pc}
&\EE_t\left[e^{-r(\tau-t)}\left((L_1\wedge S_{\tau}\m K)1_{\{\tau<T_1\}}+(L_2\wedge S_{\tau}\m K)1_{\{\tau\ge T_1\}}\right)\right]\\
&\le (L_1\m K)Q_t(\tau<T_1)+e^{-r(T_1-t)}(L_2\m K)Q_t(\tau\ge T_1)\le L_1\m K\notag
\end{align}
and therefore
the immediate exercise payoff $L_1-K$ dominates the expected value of any admissible exercise strategy.
Thus, it is optimal to exercise at once for any $t\le t^0$ and $S\ge L_1$.
 \vs{2pt}

Now, we take any $t\in(t^0,T_1)$. Then choose large enough $\wt{s}=\wt{s}(t)>0$
such that the probability $e^{r(T_1-t)}(L_1- K)/(L_2-K)\le Q_t(S_{T_1}>L_2)<1$. If we consider the strategy of exercising at $T_1$, then clearly
\begin{equation}\label{C-0}
C^{A,L}(S,t)\ge \EE_t\left[e^{-r(T_1-t)}(L_2\m K)1_{\{S_{T_1}>L_2\}}\right]>(L_2\m K)e^{-r(T_1-t)}Q_t(S_{T_1}>L_2)>L_1-K
\end{equation}
 for $S>\wt{s}$ as $t>t^0$ and hence it is not optimal to stop at $(S,t)$. Thus, we proved that there exists
$\wt{s}$ such that for any $S>\wt{s}$ it is optimal to continue at $(S,t)$.
 \vs{2pt}


$(v)$ We know the optimal exercise strategy on $[t^1,T_1]$ and proved above that points $(L_1,t)\in\mathcal{E}_{1}$ for $t\in[0,t^1]$. Now let us consider the points $(S,t)$, where $S>L_1$ and $t<t^1$, and the exercise strategy given by the first hitting time of the cap $L_1$ if it occurs before $t_1$, otherwise we follow the
already known optimal exercise rule starting from $t^1$. The value of this strategy can be computed as
\begin{equation}\label{C-0}
C^{0}(S,t)=(L_1\m K)\EE_{t}\left[ e^{-r\left( \tau_{L_1}-t\right) }1_{\{\tau_{L_1}<t^{1}\}}\right]+
\EE_{t}\left[ e^{-r\left( t^{1}-t\right) }C^{w}\left( S_{t^{1}},t^{1}\right)1_{\{\tau_{L_1}\ge t^{1}\}} \right]
\end{equation}
 where $\tau_{L_1}=\inf\{u\ge t: S_u=L_1\}$ denotes the first hitting time of $L_1$. We refer to Lemma \protect\ref{lemma:appendix} in Appendix in order to compute both expectations in \eqref{C-0}. In general, we have that there exists
 $T_0\in[0,t^1]$ (numerical results show that $T_0<t^1$, see Figure 5) such that $C^0(S,t)>L_1-K$ for $S>L_1$ and $t\in[T_0,t^1]$. Therefore, it is optimal not
 to exercise and $C^{A,L}=C^0$ for $S>L_1$ and $t\in[T_0,t^1]$.
\end{proof}

\vspace{6pt}

\begin{figure}[t]
\begin{center}
\includegraphics[scale=0.75]{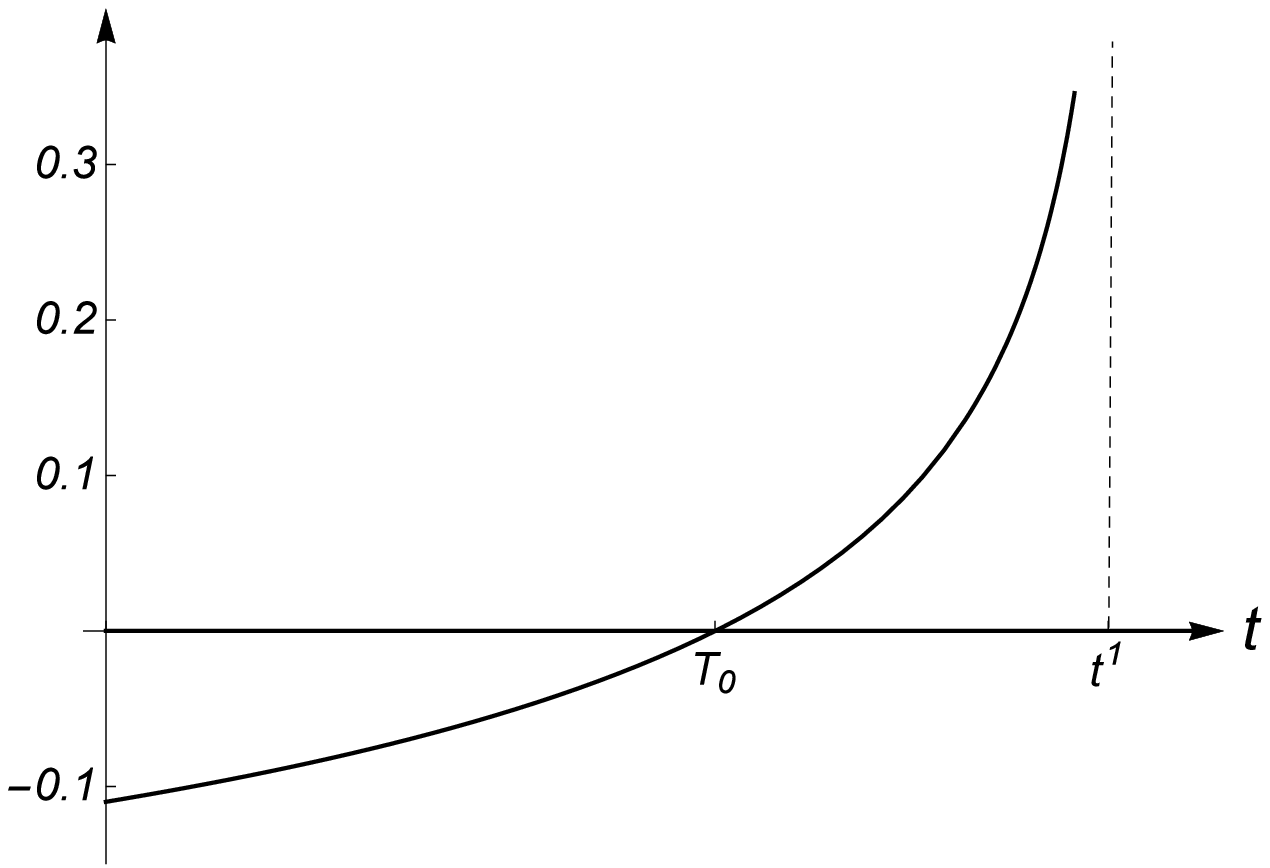}
\end{center}
\par
{}
\par
\leftskip=1.6cm \rightskip=1.6cm {\small \noindent \vspace{-10pt} }
\par
{\small \textbf{Figure 5.} This figure plots the right-hand side derivative $%
C_{S}^{0}(L_1+,t)$ of $C^{0}$ at $L_1$ for $t\in [0,t^1)$. The parameter set
is $T_1=3, T_2=4, K=1, L_1=1.3, L_2=1.39, r=0.1, \delta=0.1, \sigma=0.3$.
We obtain $T_0=1.78$. }
\par
\vspace{10pt}
\end{figure}

3. From Theorem 3.2, we have that $C^{A,L}=C^{w}$ on $\{ \left( S,t\right)
\in \mathbb{R}^{+}\times [t^{1},T_1]\}$ and it is natural to introduce the
upper exercise boundary of $\mathcal{E}_{1}$ in the case $t^1>0$.

\begin{definition}
We define the upper exercise boundary $B^{L,1}=(B^{L,1}(t))_{t\in[0,t^1]}$ such that the immediate
exercise subregion $\mathcal{E}_{1}$ can be written as
\begin{equation*}
\hspace{4pc} \mathcal{E}_{1}=\left\{ \left( S,t\right) \in \mathbb{R}%
^{+}\times \left[ 0,t^{1}\right] :L_{1}\leq S\leq B^{L,1}(t)\right\}.
\end{equation*}
\end{definition}

From Theorem 3.2, we immediately have the following properties of $B^{L,1}$:
a) $B^{L,1}(t)\le B^{w}(t)$ for $t\in[0,t^1]$; b) $B^{L,1}(t)=L_1$ for $t\in[%
T_0,t^1]$; c) $B^{L,1}(t)<+\infty$ for $t>t^0\vee 0$ and if $t^0\ge 0$ then $B^{L,1}(t)=+\infty$ for $t\in[0,t^0]$.

Using the property b) above, we can rewrite $\mathcal{E}_{1}$ as $\mathcal{%
\widetilde{E}}_{1}\cup \mathcal{\overline{E}}_{1}$ where
\begin{align}
\hspace{4pc} &\mathcal{\widetilde{E}}_{1}=\left\{ \left( S,t\right) \in
\mathbb{R}^{+}\times \left[ 0,T_{0}\right] :L_{1}\leq S\leq
B^{L,1}(t)\right\} \\
&\mathcal{\overline{E}}_{1}=\left\{ \left( S,t\right) \in \mathbb{R}%
^{+}\times \left[ T_0,t^{1}\right] :S=L_{1}\right\}.
\end{align}

4. Next, we provide lower and upper bounds for $\mathcal{E}_{1}$ using
simple benchmarks. First, let us consider the capped option price $%
C^{A,L,\infty }$ with infinite maturity $T_{2}$ and cap $L_{\tau }^{\infty
}=L_{1}1_{\tau <T_{1}}+L_{2}1_{T_{1}\leq \tau <\infty }$. Let $\mathcal{E}%
^{\infty }$ be the immediate exercise region for this option and define the
boundary $B^{w,\infty }=(B^{w,\infty }(t))_{t\in[0,T_1]}$, time $t^{1,\infty }$ and region
$\mathcal{D}^{w,\infty }$ in the same way as $B^{w}$, $t^{1}$ and  $\mathcal{D}^{w}$,
respectively. Clearly, $\mathcal{E}^{\infty }=\mathcal{E}_{1}^{\infty }\cup
\mathcal{E}_{2}^{\infty }$ where $\mathcal{E}_{1}^{\infty }\subseteq
\mathcal{D}^{w,\infty }$ satisfies properties $(ii)$-$(v)$ in Theorem \ref%
{td} and $\mathcal{E}_{2}^{\infty }=\left\{ \left( S,t\right) \in \mathbb{R}%
^{+}\times \left[ T_{1},T_{2}\right) :S\geq B^{\infty }\wedge L_{2}\right\} $%
, where $B^{\infty }$ is the optimal threshold of the uncapped perpetual call
option. Moreover, the following lemma holds.

\begin{lemma}
\label{lci}The exercise subregion $\mathcal{E}_{1}^{\infty }$ is
left-connected: if $\left( S,t\right) \in \mathcal{E}_{1}^{\infty }$ for $%
t\leq t^{1,\infty }$ and $S>L_{1}$, then $\left( S,v\right) \in \mathcal{E}%
_{1}^{\infty }$ for all $v\in \left[ 0,t\right] $.
\end{lemma}

\begin{proof} The proof is by contradiction.
Suppose that $\left( S,t\right) \in \mathcal{E}_{1}^{\infty }$ for $t\leq
t^{1,\infty}$ and $S>L_{1}$ and that there exists $v\in \left[ 0,t\right] $ such
that $\left( S,v\right) \notin \mathcal{E}_{1}^{\infty }$, i.e., $%
C^{A,L,\infty }\left( S,v\right) >$ $L_{1}-K$. Let $\tau _{v}^{\infty }$ be
the optimal exercise policy at $\left( S,v\right) $. Now consider the translated
capped contract with $L_{\tau }^{t}=L_{1}1_{\tau
<T_{1}+t-v}+L_{2}1_{T_{1}+t-v\leq \tau <\infty }$ and price $C^{A,L,\infty
}\left( S,s;T_{1}\p t\m v\right) $ for $s\in \left[ 0,\infty \right) $. Note
that $C^{A,L,\infty }( S,\cdot;T_{1}\p t\m v) \leq C^{A,L,\infty }(S,\cdot)$.
As $C^{A,L,\infty }\left( S,t;T_{1}\p t\m v\right) =C^{A,L,\infty }\left(
S,v\right) $, it then follows that $C^{A,L,\infty }\left( S,v\right) \leq
C^{A,L,\infty }\left( S,t\right) $. Hence, given the initial assumptions, $%
L_{1}-K<C^{A,L,\infty }\left( S,v\right) \leq C^{A,L,\infty }\left(
S,t\right) =L_{1}-K$, we arrive at a contradiction. It therefore must be that $\left(
S,v\right) \in \mathcal{E}_{1}^{\infty }$.
\end{proof}

If we define the upper boundary $B^{L,1,\infty }=(B^{L,1,\infty}(t))_{t\in[0,t^{1,\infty}]}$ of $\mathcal{E}%
_{1}^{\infty }$, then left-connectedness of $\mathcal{E}_{1}^{\infty }$
implies that the boundary $B^{L,1,\infty }$ is non-increasing. Moreover, as $%
C^{A,L,\infty}\ge C^{A,L}$, we have that $\mathcal{E}_{1}^{\infty }\subseteq
\mathcal{E}_{1}$ and hence the boundary $B^{L,1,\infty }$ is a lower bound
for $B^{L,1}$. \vspace{2pt}

Next, we consider the contract with restricted exercise provision, limited
to $\left[ 0,T_{1}\right] $, and automatic exercise at $L_{2}$ in $\left[
T_{1},T_{2}\right] $. Let $C^{A,L,T_{1}}\left( S,v\right) $ be the contract
price and $\mathcal{E}^{T_{1}}$ be the immediate exercise region. We define
the boundary $B^{w,T_1}=(B^{w,T_1}(t))_{t\in[0,T_1]}$ and region $\mathcal{D}^{w,T_1}$ in the same way as
in \eqref{B1} and \eqref{Dw}, respectively. Note that $\mathcal{E}^{T_{1}}=%
\mathcal{E}_{1}^{T_{1}}\cup \mathcal{E}_{2}^{T_{1}}$ where $\mathcal{E}%
_{1}^{T_{1}}$ $\subseteq \mathcal{D}^{w,T_{1}}$ satisfies properties $(ii)$-$%
(v)$ in Theorem \ref{td} and $\mathcal{E}_{2}^{T_{1}}=\left\{ \left(
S,t\right) \in \mathbb{R}^{+}\times \left[ T_{1},T_{2}\right] :S\geq
L_{2}\right\} $ is the region with automatic exercise at $L_{2}$. It is
clear that $C^{A,L,T_{1}}\left( S,t\right) \leq C^{A,L}\left( S,t\right) $
for all $t\in \left[ 0,T_{2}\right]$, so that $\mathcal{E}_{1}\subseteq
\mathcal{E}_{1}^{T_{1}}$.

To summarize, we have the following lemma.

\begin{lemma}
\label{lcau}$\mathcal{E}_{1}^{\infty }\subseteq \mathcal{E}_{1}\subseteq
\mathcal{E}_{1}^{T_{1}}$, where $\mathcal{E}_{1}^{\infty }$ and $\mathcal{E}%
_{1}^{T_{1}}$ satisfy the properties of Theorem \ref{td}.
\end{lemma}



5. We already know the optimal exercise strategy and the option price on
interval $[T_{0},T_{2}]$, so that we can set $T_{0}$ as the new maturity date
of the contract. Thus, the American capped option \eqref{problem-1} with
discontinuous cap $L_{\tau }=L_{1}1_{\left\{ \tau <T_{1}\right\}
}+L_{2}1_{\left\{ T_{1}\leq t\leq T_{2}\right\} }$ and maturity date $T_{2}$
is equivalent to an American capped derivative with maturity date $T_{0}$
and exercise payoff
\begin{equation}
\hspace{6pc}G(S_{\tau },\tau )=\left( S_{\tau }\wedge L_{1}-K\right)
^{+}1_{\left\{ \tau <T_{0}\right\} }+C^{0}\left( S_{\tau },\tau \right)
1_{\left\{ \tau =T_{0}\right\} }
\end{equation}%
at $\tau \in \left[ 0,T_{0}\right] $. The equivalent contract has a constant
cap $L_{1}$, in the interval $\left[ 0,T_{0}\right) $ and terminal value at $%
T_{0}$ given as the known value of the American capped option $%
C^{A,L}(S_{T_{0}},T_{0})=C^{0}(S_{T_{0}},T_{0})$ (see \eqref{C-0} for the
definition of $C^{0}$).

For this contract, the instantaneous benefits of waiting to exercise are
given by%
\begin{equation}
\hspace{6pc} H\left( S\right) \equiv \left( \mathbb{L}_{S}G-rG\right) \left(
S\right)  \label{H-1}
\end{equation}
for $t<T_0$ and thus
\begin{equation}
\hspace{6pc} H(S) =h_{1}( S) 1_{\{K\le S<L_1\}} +h_{2}1_{\{S\geq L_{1}\}}
\label{H-2}
\end{equation}
for $S>0$ where%
\begin{equation}
\hspace{3pc} h_{1}\left( S\right) =\left( r\! -\! \delta \right)S \! -\!
r\left( S\! -\! K\right) =rK\! -\! \delta S\text{\ \ \ \ and\ \ \ }%
h_{2}=-r\left( L_{1}\! -\! K\right).
\end{equation}

6. Standard Markovian arguments lead to the following free-boundary problem
for the value function $C^{A,L}=C^{A,L}(S,t)$ and the upper exercise
boundary $B^{L,1}=(B^{L,1}(t))_{t\in[0,T_0]}$ to be determined
\begin{align}  \hspace{5pc}  \label{PDE}
&C^{A,L}_t \! +\! \mathbb{L}_{S} C^{A,L}-rC^{A,L}=0 & \hspace{-30pt}\text{in}\; \mathcal{C}_1  \\
 \label{IS}&C^{A,L}(B^{L,1}(t)+,t)=L_1-K & \hspace{-30pt}\text{for}\; t\in[0,T_0) \\
&C^{A,L}(S,t)>G(S,t) & \hspace{-30pt}\text{in}\; \mathcal{C}_1 \\
&C^{A,L}(S,t)=G(S,t) & \hspace{-30pt}\text{in}\; \mathcal{\widetilde{E}}_1
\end{align}
where the continuation set $\mathcal{C}_1$ and the exercise region $\mathcal{%
\widetilde{E}}_1$ on $[0,T_0)$ are given by
\begin{align}  \label{D}
\hspace{5pc} &\mathcal{C}_1= \{\, (S,t)\in\mathbb{R}^{+}\! \times\!
[0,T_0]:S<L_1\;\text{or}\;S>B^{L,1}(t) \, \} \\[3pt]
&\mathcal{\widetilde{E}}_{1}=\left\{ \left( S,t\right) \in \mathbb{R}%
^{+}\times \left[ 0,T_0\right] :L_{1}\leq S\leq B^{L,1}(t)\right\}.
\end{align}

\begin{remark}
We note that the system above does not impose a smooth-fit condition at $B^{L,1}$.
 In typical stopping time problems this property follows from the fact that an underlying process immediately
enters into the exercise region if it starts just above or below the exercise
boundary. However, in our problem, the boundary $B^{L,1}$ is not always increasing.
Therefore, this standard argument cannot be applied and we are not able to prove the smooth-fit property. Intuitively, it should hold and
the numerical results support this intuition. But in order to keep the analysis as rigorous as possible, we do not assume
that the smooth-fit condition holds and thus a local time term will appear in the pricing formulas.
\end{remark}

7.  We now assume that
the technical conditions required on $B^{L,1}$ for the
local time-space formula on curves (Peskir (2005)) are satisfied, i.e., $B^{L,1}$ is continuous and of bounded variation on $[0,T_0)$. We then apply the
formula to $e^{-r(T_{0}-t)}C^{A,L}(S_{T_{0}},T_{0})$ and obtain
\begin{align}
\hspace{2.0075pc}e^{-r(T_{0}-t)}& C^{A,L}(S_{T_{0}},T_{0})  \label{ltsf} \\
=\;& C^{A,L}(S,t)+M_{T_{0}}  \notag \\
& +\int_{t}^{T_{0}}e^{-r(v-t)}(C_{t}^{A,L}+\mathbb{L}_{S}C^{A,L}-rC^{A,L})%
\left( S_{v},v\right) dv  \notag \\
& +\frac{1}{2}\int_{t}^{T_{0}}e^{-r(v-t)}\Delta
_{S}C^{A,L}(B^{L,1}(v),v)1_{\{S_{v}=B^{L,1}(v)>L_{1}\}}d\ell _{v}^{B^{L,1}}
\notag \\
& +\frac{1}{2}\int_{t}^{T_{0}}e^{-r(v-t)}\Delta
_{S}C^{A,L}(B^{L,1}(v),v)1_{\{S_{v}=B^{L,1}(v)=L_{1}\}}d\ell _{v}^{B^{L,1}}
\notag \\
& +\frac{1}{2}\int_{t}^{T_{0}}e^{-r(v-t)}\Delta
_{S}C^{A,L}(L_{1},v)1_{\{S_{v}=L_{1}<B^{L,1}(v)\}}d\ell _{v}^{L_{1}}  \notag
\\
=\;& C^{A,L}(S,t)+M_{T_{0}}+\int_{t}^{T_{0}}e^{-r(v-t)}h_{2}1_{\{S_{v}\in
(L_{1},B^{L,1}(v))\}}dv  \notag \\
& +\frac{1}{2}\int_{t}^{T_{0}}e^{-r(v-t)}C_{S}^{A,L}(B^{L,1}(v)+,v)1_{%
\{S_{v}=B^{L,1}(v)\}}d\ell _{v}^{B^{L,1}}  \notag \\
& -\frac{1}{2}\int_{t}^{T_{0}}e^{-r(v-t)}C_{S}^{A,L}(L_{1}-,v)1_{%
\{S_{v}=L_{1}\}}d\ell _{v}^{L_{1}}  \notag \\
=\;& C^{A,L}(S,t)+M_{T_{0}}+\int_{t}^{T_{0}}e^{-r(v-t)}h_{2}1_{\{S_{v}\in
(L_{1},B^{L,1}(v))\}}dv  \notag \\
& +\frac{1}{2}\int_{t}^{T_{0}}e^{-r(v-t)}C_{S}^{A,L}(B^{L,1}(v)+,v)d\ell
_{v}^{B^{L,1}}  \notag \\
& -\frac{1}{2}\int_{t}^{T_{0}}e^{-r(v-t)}C_{S}^{A,L}(L_{1}-,v)d\ell
_{v}^{L_{1}}  \notag
\end{align}%
where $\Delta _{S}C^{A,L}(S,t)\equiv C_{S}^{A,L}(S+,t)-C_{S}^{A,L}(S-,t)$ is the
jump of the derivative of $C^{A,L}$ at $S>0$ for $t\in \lbrack 0,T_{0})$; $%
\ell ^{B^{L,1}}$ and $\ell ^{L_1}$ are the local times that $S$ spends at $B^{L,1}$ and $L_1$, respectively; $M=(M_{s})_{s\geq t}$ is the
martingale term, and we exploited \eqref{H-2} with \eqref{PDE}. We also used
that $C_{S}^{A,L}(B^{L,1}(t)-,t)=0$ and $C_{S}^{A,L}(L_{1}+,t)=0$ for $t\in
\lbrack 0,T_{0})$ when $B^{L,1}(t)>L_{1}$. Note that $C^{A,L}(S,t)$ is known
for $S<L_{1}$ and $t\in \lbrack 0,T_{0})$ as the optimal exercise rule is to
wait until we hit $L_{1}$ before $T_{0}$, otherwise we obtain the value $%
C^{0}(S_{T_{0}},T_{0})$, i.e.,
\begin{equation}
C^{A,L}(S,t)=(L_{1}\!-\!K)\mathsf{E}_{t}\left[ e^{-r\left( \tau
_{L_{1}}-t\right) }1_{\{\tau _{L_{1}}<T_{0}\}}\right] +\mathsf{E}_{t}\left[
e^{-r\left( T_{0}-t\right) }C^{0}\left( S_{T^{0}},T_{0}\right) 1_{\{\tau
_{L_{1}}\geq T_{0}\}}\right]   \label{ld}
\end{equation}%
where $\tau _{L_{1}}=\inf \{u\geq t:S_{u}=L_{1}\}$ denotes the first hitting
time of $L_{1}$ when $S_{t}=S<L_{1}$ (see Lemma \protect\ref{lemma:appendix} in Appendix for the formulas for both expectations).
\vspace{2pt}

Now, taking expectation $\mathsf{E}_{t}$, using the optional sampling
theorem, the terminal condition $C^{A,L}(S,T_{0})=C^{0}(S,T_{0})$ for $S>0$,
the formula \eqref{expected-local-time-2} in terms of local times at $L_{1}$
and $B^{L,1}$, and rearranging terms in \eqref{ltsf}, we obtain the early
exercise premium representation below.

\begin{theorem}
\label{teep}The price of the American capped call option has the EEP
representation
\begin{equation}
\hspace{3pc} C^{A,L}(S,t)=C^{E,L}(S,t)+\Pi(S,t;B^{L,1}(\,\cdot\,))
\label{EEP}
\end{equation}
for $S>0$ and $t<T_0$ where
\begin{equation}\label{C-EL}
\hspace{0pc} C^{E,L}(S,t)=\;\mathsf{E}_{t}\left[ e^{-r(T_0-t)}C^{0}\left(
S_{T_0},T_0\right) \right]
\end{equation}
\vspace{-20pt}
\begin{align}\label{Pi}
\hspace{0pc} \Pi(S&,t;B^{L,1}(\,\cdot\,))= \\
=\;&r(L_1-K)\int_{t}^{T_0}e^{-r(v-t)} \mathsf{E}_t \left[ 1_{\{S_v\in
(L_{1},B^{L,1}(v))\}} \right] dv  \notag \\
&-\frac{1}{2}\int_{t}^{T_0}e^{-r(v-t)} C_S^{A,L}(B^{L,1}(v)+,v)
\;\varphi\left(-\tfrac{\log{(B^{L,1}(v)/S)}-(r-\delta-\frac{\sigma^2}{2}%
)(v-t)}{\sigma\sqrt{v-t}}\right)\frac{\sigma B^{L,1}(v)}{\sqrt{v-t}}dv
\notag \\
&+\frac{1}{2}\int_{t}^{T_0}e^{-r(v-t)}C_{S}^{A,L}(L_1-,v) \;\varphi\left(-%
\tfrac{\log{(L_1/S)}-(r-\delta-\frac{\sigma^2}{2})(v-t)}{\sigma\sqrt{v-t}}%
\right)\frac{\sigma L_1}{\sqrt{v-t}}dv.  \notag
\end{align}
In this expression, $C^{E,L}(S,t)$ is the price of the European derivative
with payoff $C^{0}(S_{T_0},T_0)$ at the maturity date $T_0$ and $%
\Pi(S,t;B^{L,1}(\,\cdot\,))$ is the early exercise premium given the optimal
exercise boundary $B^{L,1}$.
\end{theorem}

8. To characterize the optimal exercise boundary $B^{L,1}=(B^{L,1}(t))_{t\in[0,T_0]}$, we insert $%
S=B^{L,1}(t)$ for $t\in (t^0\vee 0,T_0)$ into \eqref{EEP} and use \eqref{IS} to derive the following recursive
integral equation for $B^{L,1}$
\begin{equation}  \label{IE}
\hspace{3pc} L_1-K=C^{E,L}( B^{L,1}(t),t)
+\Pi(B^{L,1}(t),t;B^{L,1}(\,\cdot\,))
\end{equation}
for $t\in (t^0\vee 0,T_0)$, subject to the boundary condition $B^{L,1}(T_0-)=L_{1}$.
This completes the proof of Theorem \ref{th:case1}$(i)$.

\end{proof}

\begin{remark}
It is important to note that the integral equation \eqref{IE} for $B^{L,1}$
is implicit, as it depends on the derivatives $C_{S}^{A,L}(L_{1}-,v)$ and $%
C_{S}^{A,L}(B^{L,1}(v)+,v)$ of $C^{A,L}$ at $L_{1}$ and $B^{L,1}$. As
mentioned above, the former derivative can be estimated independently of $%
B^{L,1}$ using \eqref{ld}. However, the latter one requires unknown values
of $C^{A,L}$ above $L_{1}$. To tackle this problem numerically, we use
backward induction with a quadrature scheme to approximate $\Pi $. Then, at
each time step $v$, we obtain the value of the boundary $B^{L,1}(v)$ as the
solution to an algebraic equation and also $C_{S}^{A,L}(B^{L,1}(v)+,v)$
using \eqref{EEP} as we already recovered $B^{L,1}(u)$ for $u>v$ (see
details in Section 7). We also note that we could avoid the presence of $%
C_{S}^{A,L}(B^{L,1}(v)+,v)$ by imposing the smooth-fit condition at $B^{L,1}$
(see Remark 4.6) and numerical results seem to provide support for this
condition. However, as we would like to keep the analysis as general and
rigorous as possible, we do not impose it.
\end{remark}

\begin{figure}[t]
\begin{center}
\includegraphics[scale=0.6]{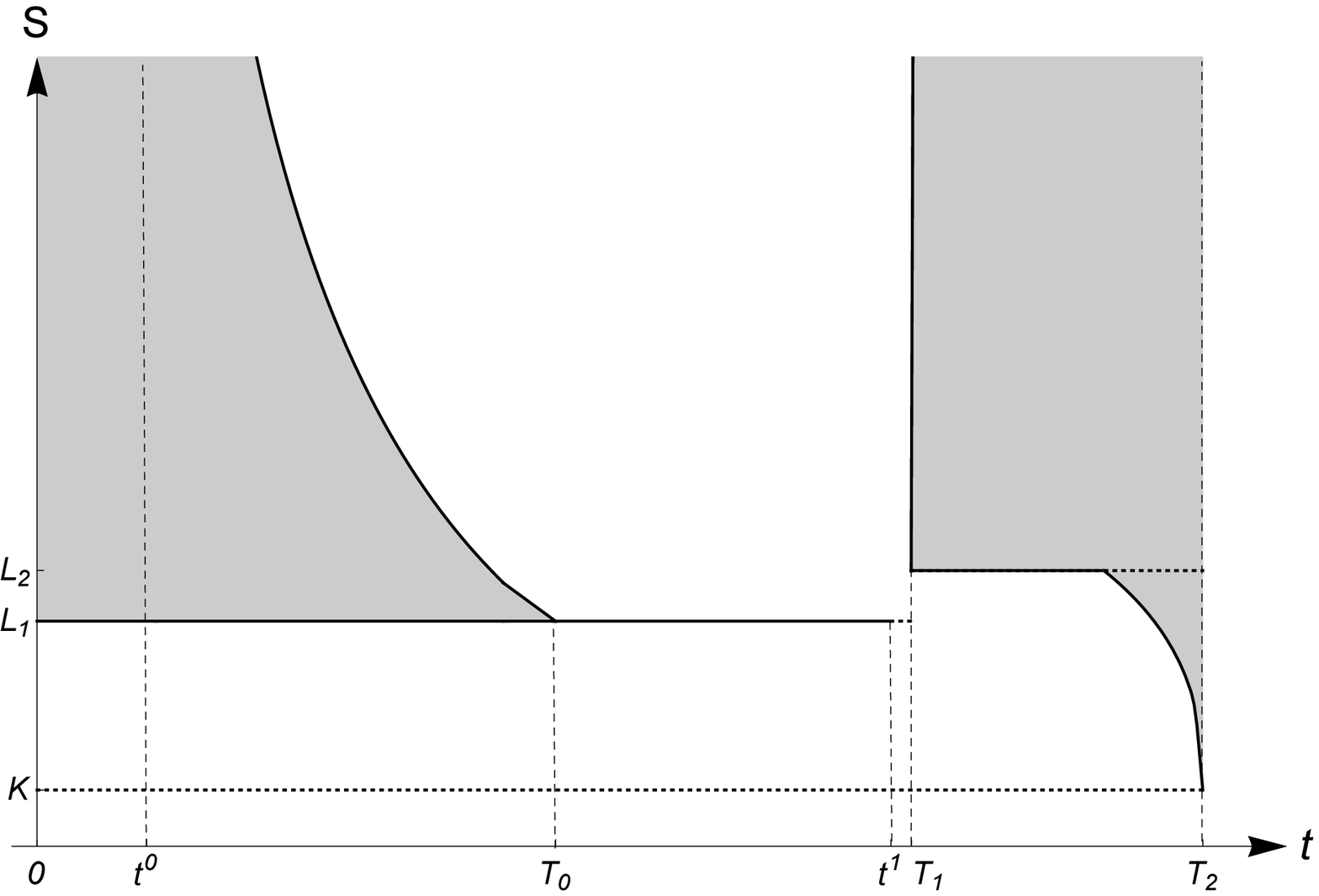}
\end{center}
\par
{}
\par
\leftskip=1.6cm \rightskip=1.6cm {\small \noindent \vspace{-10pt} }
\par
{\small \textbf{Figure 6.} This figure plots the immediate exercise region
(gray) $\mathcal{E}$ in problem \eqref{problem-1}. The part $\mathcal{E}_1$
of the exercise region on $[0,T_1)$ consists of $\mathcal{\widetilde{E}}_{1}$
(with upper boundary $B^{L,1}$ on $[0,T_0]$) and the horizontal segment $%
\mathcal{\overline{E}}_{1}=\{ \left( S,t\right) \in \mathbb{R}
^{+}\!\times\! [T_0,t^{1}] :S=L_{1}\}$. The parameter set is $T_1=3, T_2=4,
K=1, L_1=1.3, L_2=1.39, r=0.1, \delta=0.1, \sigma=0.3$. We obtain that $%
t^0=0.22$, $T_0=1.78$ and $t^1=2.93$. }
\par
\vspace{10pt}
\end{figure}

\begin{figure}[htb]
\begin{center}
\includegraphics[scale=0.7]{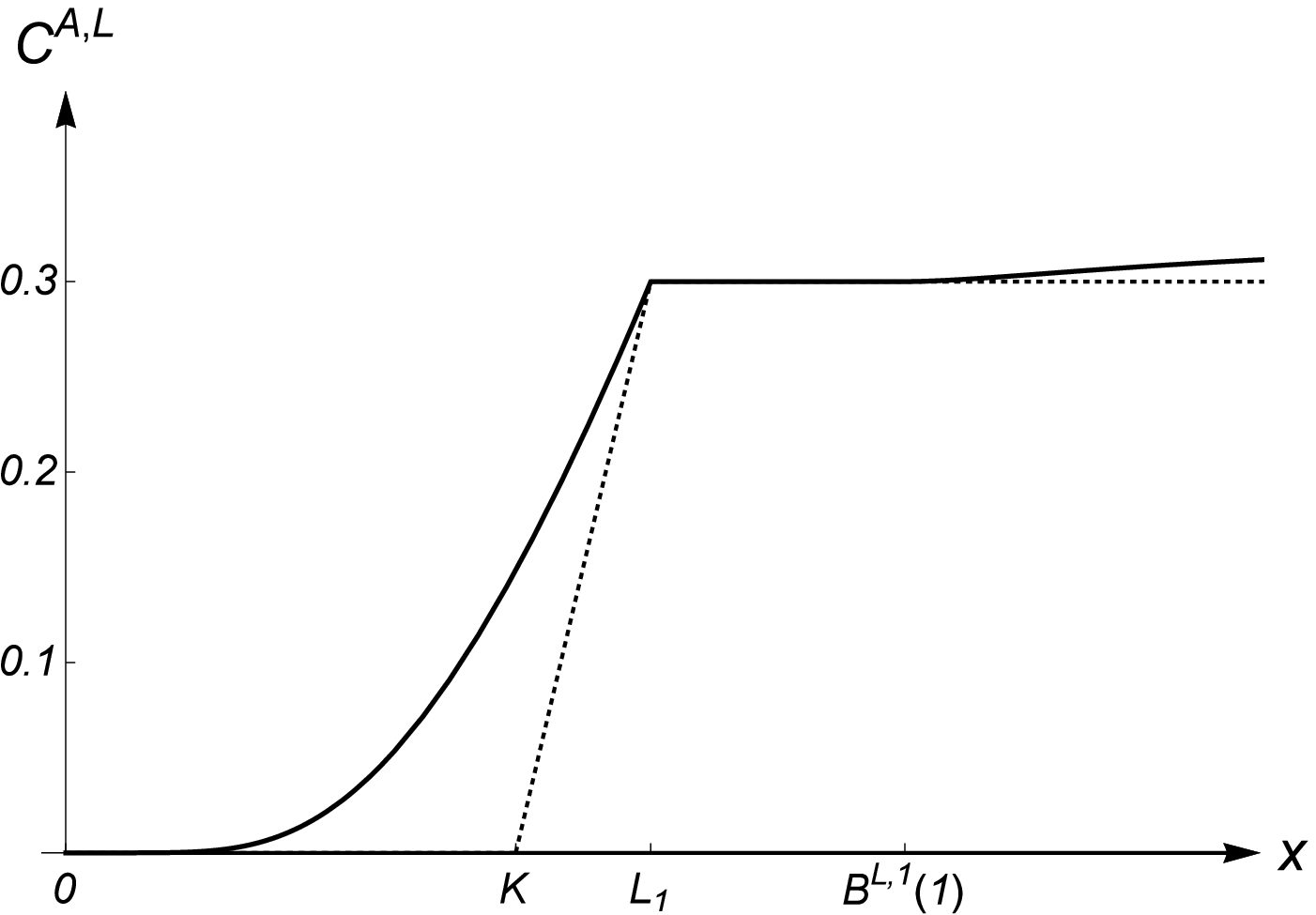}
\end{center}
\par
{}
\par
\leftskip=1.6cm \rightskip=1.6cm {\small \noindent \vspace{-10pt} }
\par
{\small \textbf{Figure 7.} This figure plots the American capped option
price $C^{A,}(S,1)$ in \eqref{problem-1}. The dotted line is the immediate
payoff $(S\wedge L_1-K)^+$. The parameter set is $T_1=3, T_2=4, K=1,
L_1=1.3, L_2=1.39, r=0.1, \delta=0.1, \sigma=0.3$. }
\par
\vspace{10pt}
\end{figure}

\begin{figure}[t]
\begin{center}
\includegraphics[scale=0.65]{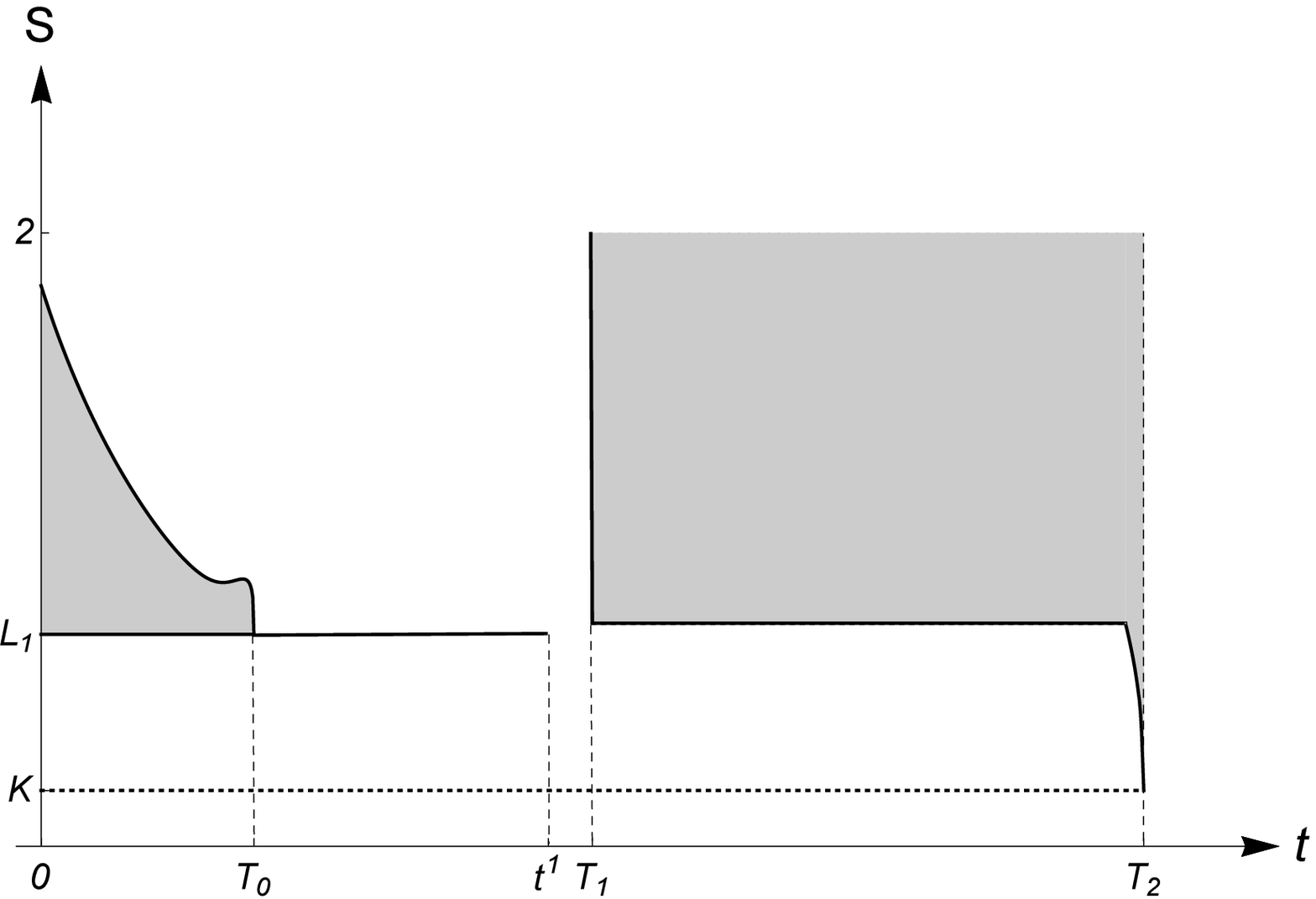}
\end{center}
\par
{}
\par
\leftskip=1.6cm \rightskip=1.6cm {\small \noindent \vspace{-10pt} }
\par
{\small \textbf{Figure 8.} This figure plots the immediate exercise region
(gray) $\mathcal{E}$ in problem \eqref{problem-1}. The part $\mathcal{E}_1$
of the exercise region on $[0,T_1)$ consists of $\mathcal{\widetilde{E}}_{1}$
(with upper boundary $B^{L,1}$ on $[0,T_0]$) and the horizontal segment $%
\mathcal{\overline{E}}_{1}=\{ \left( S,t\right) \in \mathbb{R}
^{+}\!\times\! [T_0,t^{1}] :S=L_{1}\}$. The exercise region on $[T_1,T_2)$
is characterized by $B^{L,2}$, see Theorem 3.2 $(i)$. The parameter set is $%
T_1=1, T_2=2, K=1, L_1=1.28, L_2=1.3, r=0.05, \delta=0.05, \sigma=0.5$. We
obtain that $T_0=0.386$ and $t^1=0.988$. It can be seen that the upper
boundary $B^{L,1}$ is not decreasing everywhere. }
\par
\vspace{10pt}
\end{figure}

\section{Case $B(T_1)\le L_1<L_2$}

Let us now assume that $L_{1}\geq B(T_{1})$, i.e., $t^{\ast }<T_{1}$. Note
that under this assumption we have that $B^{L,2}(t)=B(t)$ for $t\in \lbrack
T_{1},T_{2}]$. It appears that this case is easier to tackle than the one in
the previous section. The proof of this case is also divided into several
steps.

\begin{proof}[Proof of Theorem \protect\ref{th:case1}$(ii)$]

1. We start from the following observations on the structure of the exercise region on $[0,T_1]$.
If $B(t)\le S< L_1$ for $t\in[t^*,T_1)$, then it is optimal to immediately exercise the capped option at $(S,t)$. Indeed, we recall
that $C^{A,L}\le C^A$ and $C^A(S,t)=S-K$ for $S\ge B(t)$. Then we have that $S\wedge L_1-K = S-K \le C^{A,L}(S,t)\le C^A(S,t)=S-K$. Therefore,
$C^{A,L}(S,t)=S\wedge L_1 -K$ which means that $(S,t)$ belongs  to the exercise region.
\vs{2pt}

Next, it is clear that $(S,t)$ belongs to the continuation region for $S<\min (L_1,B(t))$ and $t\in[0,T_1)$. To prove it, the standard arguments with shorter maturity option can be applied
(see, e.g., the proof of Lemma 4.1 or Broadie and Detemple (1995)).
 Using exactly the same arguments as in the previous section (see the proof of Theorem 4.2$(ii)$), we can also prove that the exercise subregion is down-connected above $L_1$.
 \vs{2pt}

 Now we show that it is optimal to exercise at $(L_1,t)$ for all $t\in [0,T_1]$.
Indeed, it is already known that $(L_1,t)\in\cE$ for all $t\in [t^*,T_1]$. Now let us fix any $t<t^*$ and use the same idea as in  the proof of Theorem 4.2$(iii)$.
We choose the auxiliary capped option with new maturity $\wt{T}_1=t^*\le T_1$ for the cap $L_1$ so that $B(\wt{T}_1)=L_1$ . Obviously, the price of the option increases. Also, we have that $\wt{t}^1$ (defined as in the previous section) equals $\wt{T}_1$. Thus, using the result of Theorem 4.2$(iii)$, we know that
$(L_1,t)$ belongs to the exercise region of the auxiliary option. As the price of this option is greater than that of the original and their payoffs coincide at $(L_1,t)$, we can conclude that $(L_1,t)\in \cE_1$.
\vs{2pt}

 We then notice that  $C^{A,L}(S,t)$ coincides with the price of the capped option $C^{A,L_1}(S,t)$ with single cap $L_1$ for $S\le L_1$ and $t<T_2$. This due to the fact that both contracts have the same optimal exercise policy and the same payoff below $L_1$.
\vs{2pt}

The next observation is that there exists $T_0<T_1$ such that $(S,t)$ belongs to the continuation region for all $S>L_1$ and $t\in[T_0,T_1)$.  To show this, for $S>L_1$ and $t\in[0,T_1)$ let us consider the policy of waiting until $T_1$ unless the stock price hits the cap $L_1$. The value of this strategy is then given by
 \begin{align}\label{terminal-payoff-case2}
C^{0}(S,t)\equiv (L_1-K)\EE_t \left[e^{-r(\tau_{L_1}-t)}1_{\{\tau_{L_1}< T_{1}\}}\right]+\EE_{t}\left[ e^{-r\left( T_{1}-t\right) }(S_{T_{1}}\wedge L_2-K)1_{\{\tau_{L_1}\ge T_{1}\}} \right]
\end{align}%
  where $\tau_{L_1}$ is the first hitting time of the cap $L_1$. In fact, this is the price of the European capped barrier down-and-out call option with rebate payoff $L_1-K$ at the barrier $L_1$. Direct examination of its well-known expression shows that the slope $C^{0}_S(L_1+,t)$ at $L_1$ converges to $1$ as $t\rightarrow T_1$. Due to its continuity, for $t$ sufficiently close to $T_1$, the derivative at $L_1$ is still positive. We then define $T_0$ as the largest positive solution to the equation $C^{0}_S(L_1,t)=0$ if one exists; otherwise, $T_0=0$. As the slope is positive for $t\in(T_0,T_1)$, we have that $C^{0}(S,t)>L_1-K$ and it is optimal to wait at $(S,t)$ for $S>L_1$.

The option price at $T_0$ then reads
 \begin{align}\label{terminal-payoff-case2a}\hs{5pc}
C^{A,L}(S,T_0)=C^{A,L_1}(S,T_0)
\end{align}%
for $S\le L_1$ and
 \begin{align}\label{terminal-payoff-case2b}\hs{5pc}
C^{A,L}(S,T_0)=C^0(S,T_0)\end{align}%
for $S>L_1$ and where $\tau_{L_1}$ is the first hitting time of the cap $L_1$.
\vs{2pt}

Finally,  if $t^0\equiv T_1 -\frac{1}{r}\log((L_2-K)/(L_1-K))\ge 0$, then  $(S,t)\in%
\mathcal{E}_{1}$ for any $t\in[0,t^0]$ and $S\ge L_1$. The proof is the same as in the previous case.

\vs{6pt}

2. The results above motivate us to define the upper exercise boundary $B^{L,1}=(B^{L,1}(t))_{t\in[0,T_1)}$ such that the exercise subregion $\cE_1$ on $[0,T_1)$ is given as
\begin{align}\hs{5pc}
\mathcal{E}_{1}=&\left\{ \left( S,t\right) \in \mathbb{R}%
^{+}\times \left[ 0,T^{1}\right) :L_{1}\leq S\leq B^{L,1}(t)\right\}\\
&\cup \left\{ \left( S,t\right) \in \mathbb{R}%
^{+}\times \left[ t^*,T^{1}\right) :B(t)\leq S\leq L_{1}\right\}.\notag
\end{align}
It is clear that $B^{L,1}=L_1$ for $t\in[T_0,T_1)$ and $B^{L,1}(t)=+\infty$ for $t\in[0,t^0]$ if $t^0\ge 0$.
\vs{2pt}

As we know the optimal exercise policy on $[T_0,T_2]$, we set the new maturity date as $T_0$.
 Hence, the American capped option \eqref{problem-1}
is equivalent to an American capped derivative with maturity date $T_{0}$
and exercise payoff
\begin{equation}
\hspace{5pc}G(S_{\tau },\tau )=\left( S_{\tau }\wedge L_{1}-K\right)
^{+}1_{\left\{ \tau <T_{0}\right\} }+C^{A,L}\left(S_{T_0},T_0\right)
1_{\left\{ \tau =T_{0}\right\} }
\end{equation}%
at $\tau \in \left[ 0,T_{0}\right]$ where $C^{A,L}(S,T_0)$ is given in \eqref{terminal-payoff-case2}-\eqref{terminal-payoff-case2b}. The equivalent contract has a constant
cap $L_{1}$ on the interval $\left[ 0,T_{0}\right) $ and a known terminal value at $T_{0}$.

For this contract, the instantaneous benefits of waiting to exercise are
given by%
\begin{equation}
\hspace{5pc} H\left( S\right) \equiv \left( \mathbb{L}_{S}G-rG\right) \left(
S\right)
\end{equation}
for $t<T_0$ and for $S>0$, and thus
\begin{equation}
\hspace{5pc} H(S) =h_{1}( S) 1_{\{K\le S<L_1\}} +h_{2}1_{\{S\geq L_{1}\}}
\label{H-2-case2}
\end{equation}
 where%
\begin{equation}
\hspace{3pc} h_{1}\left( S\right) =\left( r\! -\! \delta \right)S \! -\!
r\left( S\! -\! K\right) =rK\! -\! \delta S\text{\ \ \ \ and\ \ \ }%
h_{2}=-r\left( L_{1}\! -\! K\right).
\end{equation}

Using standard arguments we obtain that the option price $C^{A,L}=C^{A,L}(S,t)$ solves PDE
\begin{align}  \hspace{5pc}  \label{PDE-case2}
C^{A,L}_t \! +\! \mathbb{L}_{S} C^{A,L}-rC^{A,L}=0
\end{align}
in the continuation set $\mathcal{C}_1= \{\, (S,t)\in\mathbb{R}^{+}\! \times\!
[0,T_0]:S<\min(L_1,B(t))\;\text{or}\;S>B^{L,1}(t) \, \}.$
\vs{6pt}

3. We apply the
local time-space formula on curves (again assuming that $B^{L,1}$ is continuous and of bounded variation on $[0,T_0]$)  to obtain
\begin{align}
\hspace{2.0075pc}e^{-r(T_{0}-t)}& C^{A,L}(S_{T_{0}},T_{0})  \label{ltsf-case2} \\
=\;& C^{A,L}(S,t)+M_{T_{0}}  \notag \\
& +\int_{t}^{T_{0}}e^{-r(v-t)}(C_{t}^{A,L}+\mathbb{L}_{S}C^{A,L}-rC^{A,L})%
\left( S_{v},v\right) dv  \notag \\
& +\frac{1}{2}\int_{t}^{T_{0}}e^{-r(v-t)}\Delta
_{S}C^{A,L}(B^{L,1}(v),v)1_{\{S_{v}=B^{L,1}(v)>L_{1}\}}d\ell _{v}^{B^{L,1}}
\notag \\
& +\frac{1}{2}\int_{t}^{T_{0}}e^{-r(v-t)}\Delta
_{S}C^{A,L}(B^{L,1}(v),v)1_{\{S_{v}=B^{L,1}(v)=L_{1}\}}d\ell _{v}^{B^{L,1}}
\notag \\
& +\frac{1}{2}\int_{t}^{T_{0}}e^{-r(v-t)}\Delta
_{S}C^{A,L}(L_{1},v)1_{\{S_{v}=L_{1}<B^{L,1}(v)\}}d\ell _{v}^{L_{1}}  \notag
\\
=\;& C^{A,L}(S,t)+M_{T_{0}}+\int_{t}^{T_{0}}e^{-r(v-t)}h_{1}(S_{v})1_{\{S_{v}\in
(B(v),L_{1})\}}dv \notag\\
&+\int_{t}^{T_{0}}e^{-r(v-t)}h_{2}1_{\{S_{v}\in
(L_{1},B^{L,1}(v))\}}dv  \notag \\
& +\frac{1}{2}\int_{t}^{T_{0}}e^{-r(v-t)}C_{S}^{A,L}(B^{L,1}(v)+,v)1_{%
\{S_{v}=B^{L,1}(v)\}}d\ell _{v}^{B^{L,1}}  \notag \\
& -\frac{1}{2}\int_{t}^{T_{0}}e^{-r(v-t)}C_{S}^{A,L}(L_{1}-,v)1_{%
\{S_{v}=L_{1}\}}d\ell _{v}^{L_{1}}  \notag \\
=\;& C^{A,L}(S,t)+M_{T_{0}}+\int_{t}^{T_{0}}e^{-r(v-t)}h_{1}(S_{v})1_{\{S_{v}\in
(B(v),L_{1})\}}dv \notag\\
&+\int_{t}^{T_{0}}e^{-r(v-t)}h_{2}1_{\{S_{v}\in
(L_{1},B^{L,1}(v))\}}dv  \notag \\
& +\frac{1}{2}\int_{t}^{T_{0}}e^{-r(v-t)}C_{S}^{A,L}(B^{L,1}(v)+,v)d\ell
_{v}^{B^{L,1}}  \notag \\
& -\frac{1}{2}\int_{t}^{T_{0}}e^{-r(v-t)}C_{S}^{A,L}(L_{1}-,v)d\ell
_{v}^{L_{1}}  \notag
\end{align}%
where $M=(M_{s})_{s\geq t}$ is the
martingale term and we used \eqref{H-2-case2} and \eqref{PDE-case2}. We also used
that $C_{S}^{A,L}(B^{L,1}(t)-,t)=0$ and $C_{S}^{A,L}(L_{1}+,t)=0$ for $t\in
\lbrack 0,T_{0})$ when $B^{L,1}(t)>L_{1}$. We recall that $C^{A,L}(S,t)=C^{A,L_1}(S,t)$
for $S<L_{1}$ and $t\in \lbrack 0,T_{1})$.

Now, taking expectation $\mathsf{E}_{t}$, using the optional sampling
theorem, the terminal condition \eqref{terminal-payoff-case2a}-\eqref{terminal-payoff-case2b},
 rearranging terms in \eqref{ltsf-case2}, we obtain the early
exercise premium representation below.

\begin{figure}[t]
\begin{center}
\includegraphics[scale=0.65]{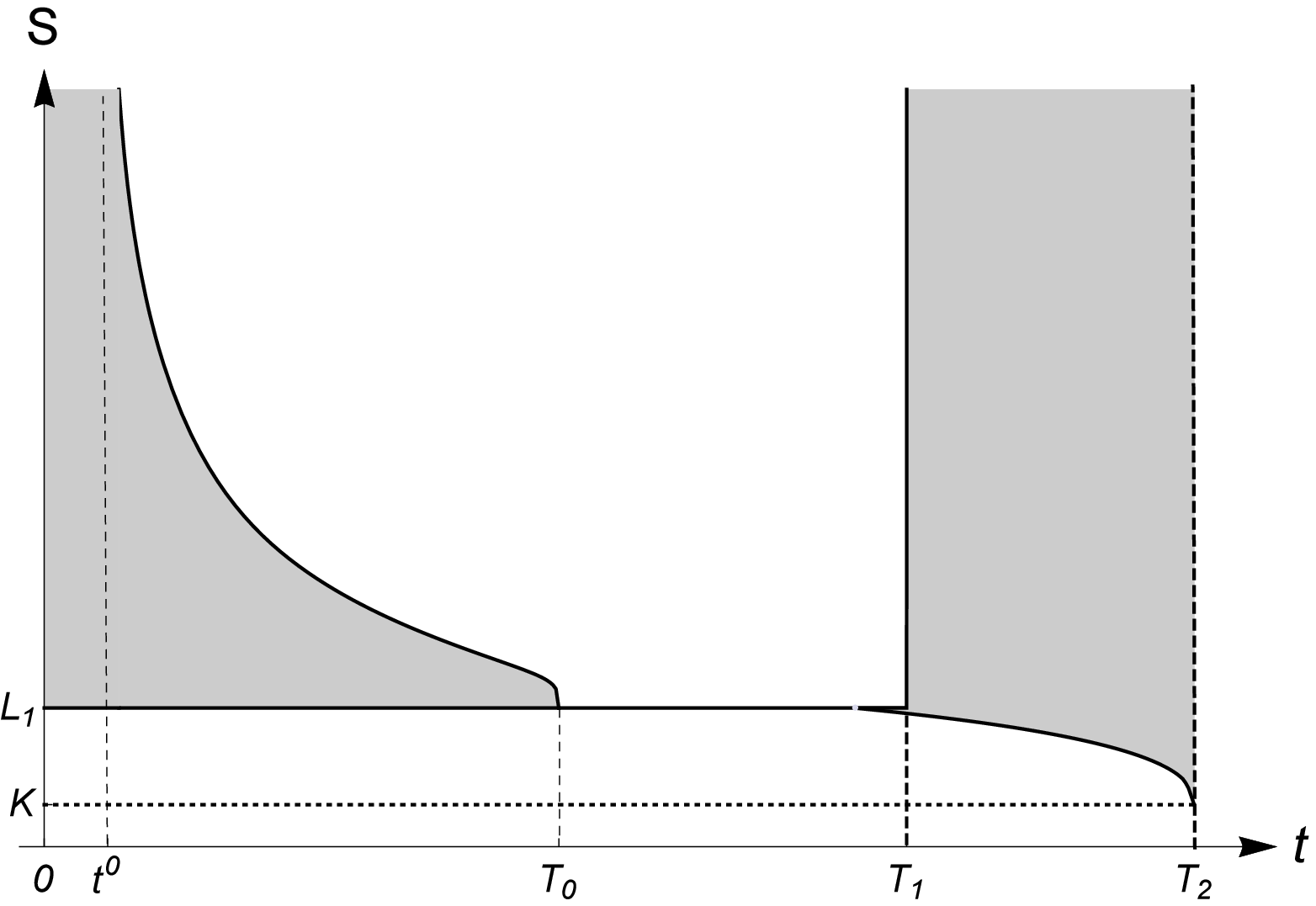}
\end{center}
\par
{}
\par
\leftskip=1.6cm \rightskip=1.6cm {\small \noindent \vspace{-10pt} }
\par
{\small \textbf{Figure 9.} This figure plots the immediate exercise region
(gray) $\mathcal{E}$ in problem \eqref{problem-1} when $B(T_1)\le L_1<L_2$.
The parameter set is $T_1=3, T_2=4,
K=1, L_1=1.46, L_2=1.5, r=0.03, \delta=0.05, \sigma=0.25$. We obtain that
$t^0=0.22$ and $T_0=1.79$. }
\par
\vspace{10pt}
\end{figure}

\begin{theorem}
\label{teep}The price of the American capped call option has the EEP
representation
\begin{equation}
\hspace{3pc} C^{A,L}(S,t)=C^{E,L}(S,t)+\Pi(S,t;B^{L,1}(\,\cdot\,))
\label{EEP-case2}
\end{equation}
for $S>0$ and $t<T_0$ where
\begin{equation}\label{C-EL-case2}
\hspace{0pc} C^{E,L}(S,t)=\;\mathsf{E}_{t}\left[ e^{-r(T_0-t)}C^{A,L}\left(
S_{T_0},T_0\right) \right]
\end{equation}
\vspace{-25pt}
\begin{align}\label{Pi-case2}
\hspace{0pc} \Pi(S&,t;B^{L,1}(\,\cdot\,))= \\
=\;&\int_{t}^{T_0}e^{-r(v-t)} \mathsf{E}_t \left[ (\delta S_v-rK)1_{\{S_v\in
(B(v),L_1)\}} \right] dv  \notag \\
&+r(L_1-K)\int_{t}^{T_0}e^{-r(v-t)} \mathsf{E}_t \left[ 1_{\{S_v\in
(L_{1},B^{L,1}(v))\}} \right] dv  \notag \\
&-\frac{1}{2}\int_{t}^{T_0}e^{-r(v-t)} C_S^{A,L}(B^{L,1}(v)+,v)
\;\varphi\left(-\tfrac{\log{(B^{L,1}(v)/S)}-(r-\delta-\frac{\sigma^2}{2}%
)(v-t)}{\sigma\sqrt{v-t}}\right)\frac{\sigma B^{L,1}(v)}{\sqrt{v-t}}dv
\notag \\
&+\frac{1}{2}\int_{t}^{T_0}e^{-r(v-t)}C_{S}^{A,L_1}(L_1-,v) \;\varphi\left(-%
\tfrac{\log{(L_1/S)}-(r-\delta-\frac{\sigma^2}{2})(v-t)}{\sigma\sqrt{v-t}}%
\right)\frac{\sigma L_1}{\sqrt{v-t}}dv.  \notag
\end{align}
In this expression, $C^{E,L}(S,t)$ is the price of the European derivative
with payoff $C^{A,L_2}(S_{T_0},T_0)$ at the maturity date $T_0$ and $%
\Pi(S,t;B^{L,1}(\,\cdot\,))$ is the early exercise premium given the optimal
exercise boundary $B^{L,1}$.
\end{theorem}

4. To characterize the optimal exercise boundary $B^{L,1}=(B^{L,1}(t))_{t\in[0,T_0]}$, we insert $%
S=B^{L,1}(t)$ for $t\in (t^0\vee 0,T_0)$ into \eqref{EEP-case2} to derive the following recursive
integral equation for $B^{L,1}$
\begin{equation}  \label{IE-case2}
\hspace{3pc} L_1-K=C^{E,L}( B^{L,1}(t),t)
+\Pi(B^{L,1}(t),t;B^{L,1}(\,\cdot\,))
\end{equation}
for $t\in (t^0\vee 0,T_0)$, subject to the boundary condition $B^{L,1}(T_0-)=L_{1}$.
This completes the proof of Theorem \ref{th:case1}$(ii)$.
\end{proof}


\section{Case $L_1>L_2$}

Finally, we consider the case $L_{1}>L_{2}$ and assume that the cap is
left-continuous
\begin{equation}
\hspace{6.1133pc}L_{\tau }=L_{1}1_{\tau \leq T_{1}}+L_{2}1_{T_{1}<\tau \leq
T_{2}}.
\end{equation}%
In this case we already know the optimal exercise rule on the interval $%
(T_{1},T_{2}]$.

\begin{proof}[Proof of Theorem \protect\ref{th:case1}$(iii)$]

1. We first note that the value of the option at $T_1$ is given as
\begin{align}\hs{5pc}
C^{A,L}(S,T_1)=G(S,T_1)\equiv \max(S\wedge L_1-K,C^{A,L_2}(S,T_1))
\end{align}
for $S>0$ as we decide whether to exercise it immediately and receive $S\wedge L_1-K$ or continue and get $C^{A,L_2}(S,T_1)$. Using that $L_1>L_2$ and the known structure of $C^{A,L_2}$ we may rewrite the option price as
\begin{align}\hs{5pc}\label{terminal-payoff-case3}
C^{A,L}(S,T_1)=(S\wedge L_1-K)1_{\{S\ge L_2\wedge B(T_1)\}}+C^{A,L_2}(S,T_1)1_{\{S<L_2\wedge B(T_1)\}}
\end{align}
for $S>0$. We also note that the payoff $G(S,t)\equiv (S\wedge L_1-K)^+$ at $t\in [0,T_1)$ is dominated by the value at $T_1$ and thus we can apply the results of Palczewski and Stettner (2010) to conclude that $C^{A,L}$ is continuous on $\mathbb{R}^+\times [0,T_1]$ and that the optimal exercise time on $[0,T_1]$ is given in the standard form
\begin{align}\hs{5pc}
\tau_*=\inf\{t>0: C^{A,L}(S_t,t)=G(S_t,t)\}.
\end{align}
One can also see that the value function $C^{A,L}(S,t)$ is discontinuous from the right at $T_1+$ for $S>L_2\wedge B(T_1)$. This is not crucial, as we concentrate on the interval $[0,T_1]$.
 \vs{6pt}

2. We now note that, in this case, the cap $L$ is non-increasing on $[0,T_2]$ and thus the option price $C^{A,L}(S,t)$ is decreasing in $t$ for fixed $S>0$. Therefore, the exercise region is right-connected. Standard dominance arguments show that the exercise region is up-connected below the cap level $L_1$ on $[0,T_1)$.
The next observation is that it is optimal to exercise at $(S,t)$ for $S\ge L_1$ and $t\in[0,T_1]$, as one attains the maximum possible payoff there.
It is also clear that the local benefits of waiting to exercise are positive when $S<rK/\delta$, so the holder should not exercise the option below $rK/\delta$ prior to $T_1$.

Consequently, we can define the non-increasing exercise boundary $B^{L,1}=(B^{L,1}(t))_{t\in[0,T_1]}$ such that
the exercise subregion $\cE_1$ on $[0,T_1]$ is given as
\begin{align}\hs{5pc}
\mathcal{E}_{1}=\left\{ \left( S,t\right) \in \mathbb{R}^{+}\times \left[ 0,T^{1}\right] : S\ge B^{L,1}(t)\right\}.
\end{align}
From the arguments above and \eqref{terminal-payoff-case3}, we have that $B^{L,1}(T_1-)=\max(rK/\delta,L_2\wedge B(T_1))$ and $B^{L,1}(T_1)=L_2\wedge B(T_1)$. Therefore the boundary $B^{L,1}$
may exhibit a jump at $T_1$.

We define time $t^*_1$ as the root of the equation $B^{L,1}(t)=L_1$ on $[0,T_1)$. If $B^{L,1}<L_1$ on $[0,T_1)$, then we let $t^*_1=0$.
As we know the optimal exercise policy on $(T_1,T_2]$ we set the new maturity date as $T_1$.
 Hence,  the American capped option \eqref{problem-1}
is equivalent to an American capped derivative with maturity date $T_{1}$
and exercise payoff
\begin{equation}
\hspace{5pc}G(S_{\tau },\tau )=\left( S_{\tau }\wedge L_{1}-K\right)
^{+}1_{\left\{ \tau <T_{1}\right\} }+G\left( S_{T_1},T_1 \right)
1_{\left\{ \tau =T_{1}\right\} }
\end{equation}%
at $\tau \in \left[ 0,T_{1}\right] $ where $G(S,T_1)$ is given in \eqref{terminal-payoff-case3}.
We define the instantaneous benefits $H(S)$ of waiting to exercise, the function $h_1 (S)$ and the constant $h_2$ as in the previous section.
Also, the option price $C^{A,L}$ solves the PDE
\begin{align}  \hspace{5pc}  \label{PDE-case3}
C^{A,L}_t \! +\! \mathbb{L}_{S} C^{A,L}-rC^{A,L}=0
\end{align}
in the continuation set $\mathcal{C}_1= \{\, (S,t)\in\mathbb{R}^{+}\! \times\!
[0,T_1):S<B^{L,1}(t) \, \}$.

We recall that the boundary $B^{L,1}$ is non-increasing and this fact allows us to easily prove that the smooth-fit condition holds at $B^{L,1}(t)$ for $t\in(t^*_1,T_1)$ (see e.g., p. 381, Section 25 in Peskir and Shiryaev (2006)). The monotonicity of $B^{L,1}$ also enables us to show the continuity of $B^{L,1}$ on $[0,T_1)$.
\vs{6pt}

\begin{figure}[t]
\begin{center}
\includegraphics[scale=0.65]{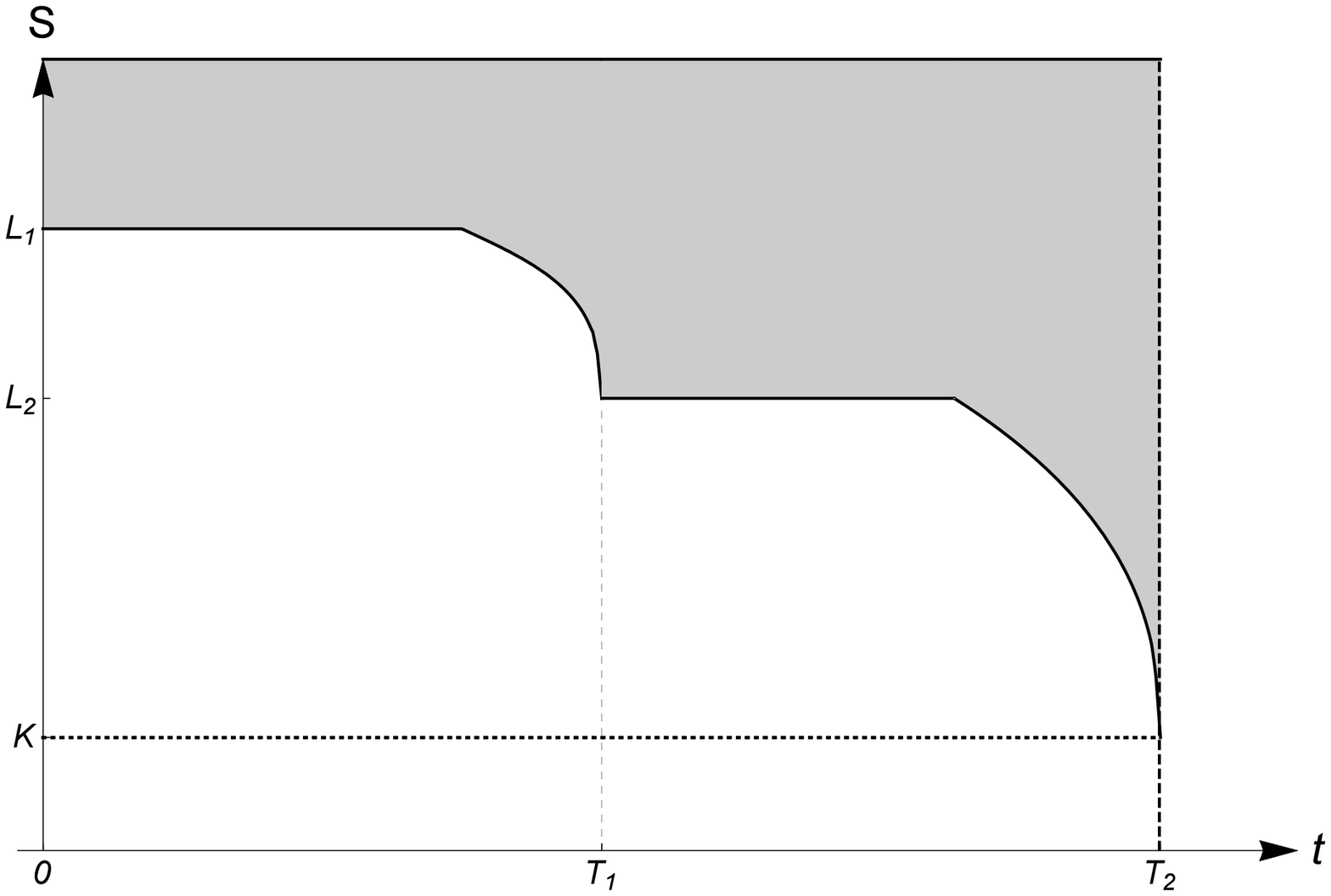}
\end{center}
\par
{}
\par
\leftskip=1.6cm \rightskip=1.6cm {\small \noindent \vspace{-10pt} }
\par
{\small \textbf{Figure 10.} This figure plots the immediate exercise region
(gray) $\mathcal{E}$ in problem \eqref{problem-1} when $L_1>L_2$.
The parameter set is $T_1=1, T_2=2,
K=1, L_1=1.45, L_2=1.3, r=0.03, \delta=0.05, \sigma=0.25$. We obtain that $%
t^*_1=0.75$ and $t^*=1.63$. }
\par
\vspace{10pt}
\end{figure}

3. In this case, we can apply the
local time-space formula on curves as the boundary $B^{L,1}$ is continuous and of bounded variation so that we obtain
\begin{align} \label{ltsf-case3} \hs{1pc}
e^{-r(T_1-t)}&C^{A,L}(S_{T_1},T_1)\\
=\;&C^{A,L}(S,t)+ \int_t^{T_1}
e^{-r(u-t)}\left(C^{A,L}_t \p\L_{S} C^{A,L}
\m rC^{A,L}\right)(S_u,u)du+M_{T_1}\nonumber\\
 &+\frac{1}{2}\int_t^{T_1}
e^{-r(u-t)}\left(C^{A,L}_S (S_u+,u)-C^{A,L}_S (S_u-,u)\right)1_{\{S_u=L_1\}}d\ell^{L_1}_u\nonumber\\
 &+\frac{1}{2}\int_{t^*_1\vee t}^{T_1}
e^{-r(u-t)}\left(C^{A,L}_S (S_u+,u)-C^{A,L}_S (S_u-,u)\right)1_{\{S_u=B^{L,1}(u)\}}d\ell^{B^{L,1}}_u\nonumber \\
=\;&C^{A,L}(S,t)+\int_{t^*_1\vee t}^{T_1} e^{-r(u-t)}(rK\m\delta S_u)1_{\{ B^{L,1}(u)\le S_u\le L_1\}}du\nonumber\\
& -r(L_1\m K)\int_t^{T_1}
e^{-r(u-t)} 1_{\{S_u\ge L_1\}}du+M_{T_1}\nonumber\\
&-\frac{1}{2}\int_t^{T_1}
e^{-r(u-t)}C^{A,L}_S (L_1-,u)d\ell^{L_1}_u\nonumber
 \end{align}
where $M=(M_t)_{t\ge 0}$ is the martingale part
and we used the PDE \eqref{PDE-case3}, the smooth-fit condition at $B^{L}$ on $(t^*_1,T_1)$,
and that $C^{A,L}_S(L_1+,\cdot)=0$.

Now, taking expectation $\mathsf{E}_{t}$, using the optional sampling
theorem, the terminal condition $C^{A,L}(S,T_{1})=G(S,T_{1})$ for $S>0$,
 rearranging terms in \eqref{ltsf-case3}, we obtain the early
exercise premium representation below.

\begin{theorem}
\label{teep}The price of the American capped call option has the EEP
representation
\begin{equation}
\hspace{3pc} C^{A,L}(S,t)=C^{E,L}(S,t)+\Pi(S,t;B^{L,1}(\,\cdot\,))
\label{EEP-case3}
\end{equation}
for $S>0$ and $t<T_1$ where
\begin{equation}\label{C-EL-case3}
\hspace{3pc}
C^{E,L}(S,t)=\;\mathsf{E}_{t}\left[ e^{-r(T_1-t)}G\left(S_{T_1},T_1\right) \right]
\end{equation}
\vspace{-25pt}
\begin{align}\label{Pi-case3}
\hspace{3pc} \Pi(S&,t;B^{L,1}(\,\cdot\,))= \\
=\;&\int_{t\vee t^*_1}^{T_1}e^{-r(u-t)} \mathsf{E}_t \left[ (\delta S_u-rK)1_{\{S_u\in
(B^{L,1}(u),L_1)\}} \right] du  \notag \\
&+r(L_1-K)\int_{t}^{T_1}e^{-r(u-t)} \mathsf{E}_t \left[ 1_{\{S_u\ge L_1\}} \right] du  \notag \\
&+\frac{1}{2}\int_{t}^{T_1}e^{-r(u-t)}C_{S}^{A,L}(L_1-,u) \;\varphi\left(-%
\tfrac{\log{(L_1/S)}-(r-\delta-\frac{\sigma^2}{2})(u-t)}{\sigma\sqrt{u-t}}%
\right)\frac{\sigma L_1}{\sqrt{u-t}}du.  \notag
\end{align}
In this expression, $C^{E,L}(S,t)$ is the price of the European derivative
with payoff $G(S_{T_1},T_1)$ at the maturity date $T_1$ and $%
\Pi(S,t;B^{L,1}(\,\cdot\,))$ is the early exercise premium given the optimal
exercise boundary $B^{L,1}$. We note that $C_{S}^{A,L}(L_1-,u)=1$ on $ [t^*_1,T_1)$.
\end{theorem}

4. To characterize the optimal exercise boundary $B^{L,1}=(B^{L,1}(t))_{t\in[0,T_1]}$, we insert
$S=B^{L,1}(t)$ into \eqref{EEP-case3} to derive the following recursive
integral equation for $B^{L,1}$
\begin{equation}  \label{IE-case3}
\hspace{3pc} B^{L,1}(t)-K=C^{E,L}( B^{L,1}(t),t)
+\Pi(B^{L,1}(t),t;B^{L,1}(\,\cdot\,))
\end{equation}
for $t\in [t^*_1,T_1)$, subject to the boundary condition $B^{L,1}(T_1-)=\max(rK/\delta,L_2\wedge B(T_1))$.
\end{proof}

Let us now briefly discuss the case of right-continuous cap $L$ with $%
L_{1}>L_{2}$. In this instance, we cannot employ the results of Palczewski
and Stettner (2010) as the payoff jumps down at $T_{1}$. In fact, in some of
the possible cases there is no optimal stopping time. For example, if $%
\delta $ is sufficiently small, then $rK/\delta >L_{2}$ and it is optimal to
wait for $S\in (L_{2},rK/\delta \wedge L_{1})$ and $t<T_{1}$ as the local
benefits of waiting are positive. However, at $t=T_{1}$, the payoff suddenly
decreases from $S-K$ to $L_{2}-K$. The latter is the maximum possible payoff
on $[T_{1}.T_{2}]$. Thus, there is no optimal exercise rule in this case.

Standard arguments show that the values of both option versions
(right-continuous and left-continuous caps) coincide at $t<T_{1}$ and $%
t>T_{1}$ for given $S>0$. The differences between two contracts are that (i)
the right-continuous version may not have an optimal exercise time; (ii) the
values at $T_{1}$ differ; (iii) the value function of the right-continuous
version has a discontinuity from the left and that of the left-continuous
version a discontinuity from the right at $T_{1}$. We also note that even
though the right-continuous version does not have an optimal exercise time
in some cases, one can construct a sequence of stopping times such that the
corresponding value functions converge to the option price when $t<T_{1}$.

\section{Numerical Algorithm}

In this section, we describe the numerical solution to the American capped
option problem \eqref{problem-1} for all three cases: (i) $%
L_1<\min(L_2,B(T_1))$; (ii) $B(T_1)\le L_1<L_2$; (iii) $L_1>L_2$. \vspace{2pt%
}

(i) Here we discuss the numerical algorithm for the case $%
L_1<\min(L_2,B(T_1))$ using \eqref{EEP} and \eqref{IE}, and provide a
step-by-step algorithm. \vspace{2pt}

\emph{Step 1}. We solve the well-known recursive integral equation for the
boundary $B$ of the uncapped call option (for details of the numerical
implementation, see, e.g., Chapter 8 of Detemple (2006)). This gives the
early exercise premium representation for the uncapped option price $%
C^{A}(S,t)$. \vspace{2pt}

\emph{Step 2}. Using the formula \eqref{capped-price-1}, Lemma \protect\ref{lemma:appendix} and the EEP formula
for $C^A$,  we compute $C^{A,L_2}(L_2\! -\! \varepsilon,t)$ for any $%
t\in[0,t^*)$ and fixed small $\varepsilon>0$ to estimate $%
C^{A,L_{2}}_S(L_2-,t)$ (see Figure 3). We then exploit \eqref{capped-price-3}
to compute $C^{A,L_2}(S,T_1)$ for any $S>0$ (see Figure 2). It is more
efficient to use \eqref{capped-price-3} and compute one-dimensional
integrals rather than the formula \eqref{capped-price-1}, as the latter
requires the integration of the call option price $C^A$, which is also given
in an integral form. \vspace{2pt}

\emph{Step 3}. Here we compute the function $C^{w}(S,t)$ for $S>0$ and $%
t<T_{1}$ by using \eqref{Cw}. \vspace{2pt}

\emph{Step 4}. We are now ready to compute $B^{w}(t)$ for all $t\in \lbrack
0,T_{1}]$ using the function $C^{w}$ from \emph{Step 3} and the definition %
\eqref{B1}, and also to determine numerically $t^{1}$ (see Figure 4).
\vspace{2pt}

\emph{Step 5}. Here, we identify $T_{0}$. For this, we compute $%
C^{0}(L_{1}\! +\! \varepsilon ,t)$ for very small $\varepsilon >0$ and $%
t<t^{1} $ using \eqref{C-0}, Lemma \protect\ref{lemma:appendix} and the values of $C^{w}$ at $t^{1}$ (see \emph{%
Step 3}) to estimate $C^{0}_S(L_{1}+,t)$ (see Figure 5). Starting from $%
t=t^{1}$, we go backward in time and observe that $C^{0}(L_{1}\! +\!
\varepsilon ,t)>L_{1}-K$ if and only if $t>T_{0}$ for some $T_{0}<t^{1}$.
Thus, we obtain $T_{0}$. \vspace{2pt}

\emph{Step 6}. Before solving the recursive integral equation \eqref{IE} for
$B^{L,1}$, we estimate $C_{S}^{A,L}( L_1-,t)$ for $t\in[0,T_0)$ as $\frac{1}{%
\varepsilon}(L_1-K-C^{A,L}( L_1-\varepsilon,t))$ for fixed small $%
\varepsilon>0$, where we use \eqref{ld}, Lemma \protect\ref{lemma:appendix} and values of $C^0$ to compute $%
C^{A,L}( L_1\! -\! \varepsilon,t)$. \vspace{2pt}

\emph{Step 7}. Finally, we have all ingredients to solve the integral
equation \eqref{IE} numerically. We discretize the interval $[0,T_{0}]$ with
step $h$ and approximate the integrals by a quadrature scheme. Then, using
the terminal condition $B^{L,1}(T_{0})=L_{1}$ and backward induction, we
recover the boundary $B^{L,1}$ at points $T_{0}-h,T_{0}-2h,...,0$ by solving
the sequence of algebraic equations. We also have to obtain $%
C_{S}^{A,L}(B^{L,1}(T_{0}-ih)+,T_{0}-ih)$ at each time step, as it is
required to compute $B^{L,1}$. The derivative is estimated from \eqref{EEP}
as we already computed $B^{L,1}(T_{0}-jh)$ and $%
C_{S}^{A,L}(B^{L,1}(T_{0}-jh)+,T_{0}-jh)$ for $1\leq j<i$. Hence, we also
need to estimate $C_{S}^{A,L}(L_1+,T_{0})$, which simply equals $%
C_{S}^{0}(L_1+,T_{0})$. Having interpolated and computed $B^{L,1}$ on $%
[0,T_{0}]$ (see Figure 6), we can recover the American capped option price $%
C^{A,L}$ using \eqref{EEP} (see Figure 7). \vspace{3pt}

\emph{Numerical examples}.

\emph{Example 1}. To illustrate the algorithm, we choose a parameters set
and implement all the steps above to obtain the immediate exercise region
and the price of the American capped call option. The parameters values are:
$T_1=3, T_2=4, K=1, L_1=1.3, L_2=1.39, r=0.1, \delta=0.1, \sigma=0.3$. For
this set, we have that $T_0=1.78, t^1=2.93$ and $t^*=3.66$. We refer to
Figures 1-7 for the graphical illustrations. Numerical experiments seem to
provide evidence that the smooth pasting condition at $B^{L,1}$ holds, so
that $C_{S}^{A,L}(B^{L,1}(v)+,v)=0$, and for this parameters set, the
boundary $B^{L,1}$ is decreasing.

\emph{Example 2}. An interesting issue in this problem is the monotonicity
of the upper boundary $B^{L,1}$. There are two opposite factors that affect
it: non-decreasing payoff and decreasing time horizon when $t$ increases.
This tradeoff prevents us from proving the monotonicity of $B^{L,1}$. In
fact, for the set of parameters in Figure 8, we obtain that the boundary is
non-monotone with some local periods of increasing behavior. This is another
striking feature of the optimal stopping problem \eqref{problem-1}. \vspace{%
6pt}

(ii) Now we turn to the case $B(T_1)\le L_1<L_2$. As we saw in Section 5,
the structure of the exercise region is simpler than in the case $%
L_1<\min(L_2,B(T_1))$ as $t^1=T_1$ and hence the numerical algorithm is less
involved. In fact, we can skip \emph{Steps 3-4} from the algorithm above and
only perform \emph{Steps 1,2,5,6 and 7} to solve \eqref{IE-case2}. We
illustrate this case in Figure 9 where we plot the immediate exercise region
$\mathcal{E}$.

\vspace{6pt}

(iii) Finally, we discuss the case $L_{1}>L_{2}$. As the exercise region is
up-connected, the numerical algorithm becomes quite simple. The approach is
similar to the algorithm for the classical uncapped call option boundary in
\emph{Step 1}. We simply discretize the interval $[0,T_{1}]$, then starting
from $B^{L,1}(T_{1}-)=\max (rK/\delta ,L_{2}\wedge B(T_{1}))$ we use
backward induction and solve the sequence of algebraic equations to recover
the exercise boundary $B^{L,1}$ on $[0,T_{1}]$. We plot the immediate
exercise region $\mathcal{E}$ in Figure 10.

\vspace{6pt}

\section{Appendix}

For the numerical computation of certain expectations appearing in the pricing formulas, we use the following auxiliary lemma.

\begin{lemma}\label{lemma:appendix}
 Let us fix some $0\le t<T$ and constant level $L>0$, and consider the first hitting time $\tau_{L}=\inf \{u\geq t:S_{u}=L\}$
of the level $L$ when the stock price starts at $S_t=S$ at time $t$. Then the following expectations can be computed as 
  \begin{equation*}
\hspace{0pc}\EE_t\left[e^{-r(\tau_L-t)}1_{\{\tau_L<T\}}\right] =\left\{
\begin{array}{l}
\lambda^{2\phi/\sigma^2}\Phi(d_0)
   +\lambda^{2\alpha/\sigma^2}\Phi(d_0+2f\sqrt{T-t}/\sigma) \quad\text{for}\; S<L\\[3pt]
\lambda^{2\phi/\sigma^2}\Phi(-d_0)
   +\lambda^{2\alpha/\sigma^2}\Phi(-d_0-2f\sqrt{T-t}/\sigma) \quad\text{for}\; S>L%
\end{array}%
\right.
\end{equation*}%
  where 
  \begin{align*}
  d_0=\frac{\log\lambda-f(T-t)}{\sigma\sqrt{T-t}}
 \end{align*}
 and $f=\sqrt{b^2+2r\sigma^2}$, $b=-(r-\delta-\sigma^2/2)$, $\phi=(b-f)/2$, $\alpha=(b+f)/2$ and $\lambda=S/L$.
\vs{2pt}

Now if we take some integrable function $G$, then we compute the following expectations as
 \begin{equation*}
\hspace{0pc}\EE_t\left[e^{-r(\tau_L\wedge T-t)}G(S_T)1_{\{\tau_L\ge T\}}\right] =\left\{
\begin{array}{l}
 e^{-r(T-t)}\int_0^L G(x)u(x,t,T)dx \quad\text{for}\; S<L\\[3pt] 
  e^{-r(T-t)}\int_L^{+\infty} G(x)u(x,t,T)dx \quad\text{for}\; S>L
\end{array}%
\right.
\end{equation*}%
 where 
  \begin{align*}
 &u(x,t,T)=\frac{\varphi(d^-(x))-\lambda^{1-2(r-\delta)/\sigma^2}\varphi(d^+(x))}{x\sigma\sqrt{T-t}}\\
 &d^{\pm}(x)=\frac{\pm\log\lambda+\log(x/L)+b(T-t)}{\sigma\sqrt{T-t}}.
 \end{align*}
 
 \begin{proof}
 The proof can be found in the Appendix of Broadie and Detemple (1995).
 \end{proof}

\end{lemma}


\begin{thebibliography}{99}
\bibitem{PS} \textsc{Boyle, P. \emph{and} Turnbull, S.} (1989) Pricing and
hedging capped options. \emph{J. Futures Markets} 9 (41--54).

\bibitem{PS} \textsc{Broadie, M. \emph{and} Detemple, J.} (1995) American
capped call options on dividend-paying assets. \emph{Rev. Financ. Stud.} 8
(161--191).

\bibitem{Chance} \textsc{Chance, D.} (1994). The pricing and hedging of
limited exercise caps and spreads. \emph{Journal of Financial Research} 17
(561--584).

\bibitem{Det} \textsc{Detemple, J.} (2006). \emph{American-Style Derivatives.%
} Chapman \& Hall/CRC.

\bibitem{Flesacker} \textsc{Flesacker, B.} (1992). The design and valuation
of capped stock index options. Working Paper, University of Illinois at
Urbana-Champaign.

\bibitem{IK} \textsc{Karatzas, I.} (1988). On the pricing of American
options. \emph{Appl. Math. Optim.} 17 (37--60).

\bibitem{Pa-St} \textsc{Palczewski, J. \emph{and} Stettner, L.} (2010).
Finite horizon optimal stopping of discontinuous functionals with
applications to impulse control with delay. \emph{SIAM J. Control Optim.} 48
(4874--4909).

\bibitem{Pe-1} \textsc{Peskir, G.} (2005). A change-of-variable formula with
local time on curves. \emph{J. Theoret. Probab.} 18 (499--535).

\bibitem{PS} \textsc{Peskir, G. \emph{and} Shiryaev, A. N.} (2006). \emph{%
Optimal Stopping and Free-Boundary Problems.} Lectures in Mathematics, ETH
Z\"urich, Birkh\"auser.

\bibitem{Qiu} \textsc{Qiu, S.} (2016). American eagle options. \emph{%
Research Report No. 1 (2016), Probab. Statist. Group Manchester} (45 pp).
\end{thebibliography}
\end{document}